\documentclass[12pt]{article}
\usepackage{amsmath}
\usepackage{graphicx}
\usepackage{natbib}
  \usepackage{lineno}

\usepackage{url} % not crucial - just used below for the URL 
\usepackage[english]{babel}
\usepackage{fontenc}
\usepackage{amssymb}
\usepackage{dsfont}
\usepackage{inputenc}
\usepackage{amsthm}
\usepackage{amssymb}
\usepackage[toc,page]{appendix}
\usepackage{titlesec}
\usepackage{xcolor}
\usepackage{enumitem}
\usepackage{tabularx}
\usepackage{algorithm,algorithmicx}
\usepackage{algcompatible}
\usepackage{subcaption}
\usepackage{hyperref}
\usepackage{comment}
\usepackage{placeins}
\usepackage{csquotes}
\usepackage{nomencl}

\newcommand{\blind}{0}

\DeclareMathOperator*{\argmax}{arg\,max}
\DeclareMathOperator*{\argmin}{arg\,min}

\newtheorem{theorem}{Theorem}[section]

\newtheorem{proposition}[theorem]{Proposition}

\newcommand{\Ecal}{\mathcal{Y}}
\newcommand{\Esp}[1]{\mathbb{E}\left[#1\right]}

\newcommand{\centro}{\gamma}
\newcommand{\Gamm}{\Gamma_{\card}}
\newcommand{\norm}[1]{\lVert #1 \lVert_{\Ecal}}
\newcommand{\rep}{q_{\Gamm}}
\newcommand{\qgs}{q_{{\Gamm}^{\star}}}

\newcommand{\scalprodE}[2]{\langle #1,#2\rangle_{\Ecal}}
\newcommand{\scalprod}[2]{\Esp{\scalprodE{#1}{#2}}}
\newcommand{\scalprodomega}[2]{\Esp{\scalprodE{#1}{#2}}}
\newcommand{\biais}{g}

\newcommand{\var}{\mathbb{V}}

\newcommand{\clust}[2]{C^{#1}_{#2}}
\newcommand{\yinclust}[3]{#1 \in \clust{#2}{#3}}
\newcommand{\card}{\ell}
\newcommand{\campbell}{h}

\newcommand{\Ytrue}{{y}}
\newcommand{\Ypred}{\hat{y}}

\newcommand{\dime}{64}
\newcommand{\Gammaref}{\Gamm^{\star}}
\newcommand{\Gammapred}{\hat\Gamm^{\star}}

\newcommand{\erreurprobarelative}{\epsilon_P^{MM}}
\newcommand{\nbgamma}{100}

\newcommand{\ntrue}{n_{\mathrm{maps}}}
\newcommand{\npred}{n_{\mathrm{maps}}}

\newcommand{\Dspace}{\mathcal{X}}
\newcommand{\supp}{{\mathrm{supp}}}
\newcommand{\thresh}{t}
\newcommand{\estim}{\hat{E}}

\newcommand{\ntrain}{n_{\mathrm{train}}}

\newcommand{\eqdef}{:=}
\newcommand{\Ypca}{y^{\mathrm{pca}}}
\newcommand{\Yproj}{y^{\mathrm{proj}}}
\newcommand{\nvar}{n_{v}}
\newcommand{\nerr}{n_{e}}
\newcommand{\npc}{n_{\mathrm{pc}}}

\makenomenclature

\begin{document}

\def\spacingset#1{\renewcommand{\baselinestretch}%
{#1}\small\normalsize} \spacingset{1}

%%%%%%%%%%%%%%%%%%%%%%%%%%%%%%%%%%%%%%%%%%%%%%%%%%%%%%%%%%%%%%%%%%%%%%%%%%%%%%

\if0\blind
{
  \title{\bf 
  Quantizing rare random maps: application to flooding visualization}
  \author{Charlie Sire \\
    \small IRSN, BRGM, Mines Saint-Etienne, Univ. Clermont Auvergne, CNRS, UMR 6158 LIMOS\\
    and \\
    Rodolphe Le Riche \\
    \small Mines Saint-Etienne, Univ. Clermont Auvergne, \small CNRS, UMR 6158 LIMOS\\
    and \\
     Didier Rullière \\
    \small Mines Saint-Etienne, Univ. Clermont Auvergne, CNRS, UMR 6158 LIMOS\\
    and \\
    Jérémy Rohmer \\
    \small BRGM\\
    and\\
     Lucie Pheulpin \\
    \small IRSN \\
    and\\
    Yann Richet \\
    \small IRSN \\
    }
    \date{}

  \maketitle
} \fi
\if1\blind
{
  \bigskip
  \bigskip
  \bigskip
  \begin{center}
    {\LARGE\bf Title}
\end{center}
  \medskip
} \fi

\bigskip
\begin{abstract}
Visualization is an essential operation when assessing the risk of rare events such as coastal or river floodings. 
The goal is to display a few prototype events that best represent the probability law of the observed phenomenon, a task known as quantization.
It becomes a challenge when data is expensive to generate and critical events are scarce, like extreme natural hazard. In the case of floodings, each event relies on an expensive-to-evaluate hydraulic simulator which takes as inputs offshore meteo-oceanic conditions and dyke breach parameters to compute the water level map. 
In this article, Lloyd's algorithm, which classically serves to quantize data, is adapted to the context of rare and costly-to-observe events.
Low probability is treated through importance sampling, while Functional Principal Component Analysis combined with a Gaussian process deal with the costly hydraulic simulations. 
The calculated prototype maps represent the probability distribution of the flooding events in a minimal expected distance sense,
and each is associated to a probability mass.
The method is first validated using a 2D analytical model and then applied to a real coastal flooding 
scenario.  
The two sources of error, the metamodel and the importance sampling, are evaluated to quantify the precision of the method.
\end{abstract}

\noindent%
\vfill

\newpage
\spacingset{1.5} % DON'T change the spacing!
\section{Introduction}

Political authorities typically support their decisions regarding flood risk management based of flood maps calculated with hydraulic simulations.
\citealt{Beven} provide a realistic application, and a review of the existing methods for graphical representations of flooding uncertainty is given by \citealt{Seipel}.

The design and analysis of numerical experiments is a generic procedure to capture and explore the uncertainties related to such heavy numerical models (\citealt{Santner}). 
It consists in performing a series of runs of the hydrodynamic flow simulator 
where the forcing conditions (the river flow rate and the offshore wave characteristics), 
the breach characteristics (location and dimension of the breach) or the embankments characteristics (geometry, failures conditions, etc.) are one realization of a random law. 
\citealt{Rohmer} provide an example of the setting up of the forcing scenarios in coastal simulations and specify a method to determine the probabilistic laws describing the natural phenomena. 

Visualizing simulations for risk assessment is a real challenge because of their time-consuming aspect, the low probability of the flooding event, and the need to communicate the results to authorities who do not necessarily have a scientific background.
To this end, a simple-but-efficient approach consists in providing a small size set of representative flood maps. 
Such prototype maps (\citealt[chap. 13]{Hastie}) are chosen to best represent the probability law associated to the flooding event. 
Calculating the prototype maps is a quantization task (\citealt{Gray}), implemented here through the Lloyd's algorithm (\citealt{Pages}). 
Lloyd's algorithm searches for Voronoi cells centroids such that the mean distance (with respect to the uncertainty distribution) between an observation and its nearest centroid is minimized. 

This procedure yields results that have both a probabilistic meaning and can be visually interpreted by a wide audience. For instance, in the flooding case tackled in this article, the four non-empty prototype maps and their probability masses expressed as frequencies are given in Figure~\ref{example_quanti}.
The empty prototype map occurs with $97\%$ probability.

\begin{figure*}[hbt!]
\centering
  \begin{subfigure}[t]{0.38\textwidth}
\includegraphics[width=1\linewidth]{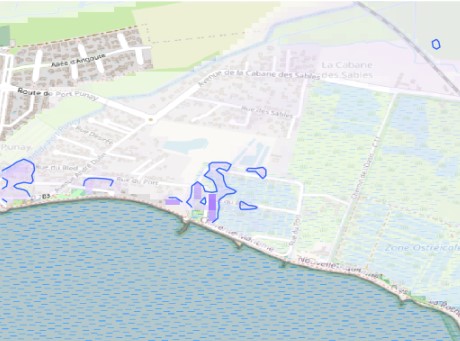}
\caption{$2.0\times10^{-2}$, 1 in 50}
 \end{subfigure}%
 ~ 
\begin{subfigure}[t]{0.38\textwidth}
\includegraphics[width=1\linewidth]{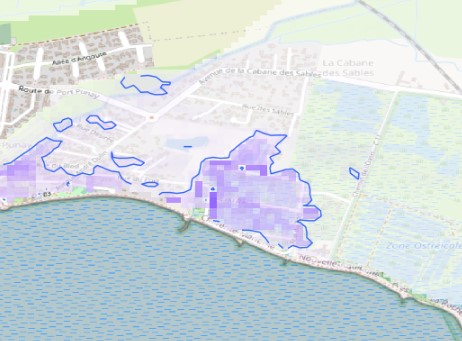}
\caption{$5.2\times 10^{-2}$,
1 in 192}
  \end{subfigure}
\begin{subfigure}[t]{0.38\textwidth}
\includegraphics[width=1\linewidth]{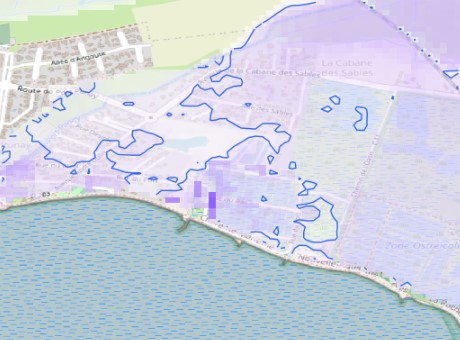}
\caption{$6.2\times 10^{-3}$, 1 in 161}
\label{clust4}
  \end{subfigure}
\begin{subfigure}[t]{0.38\textwidth}
\includegraphics[width=1\linewidth]{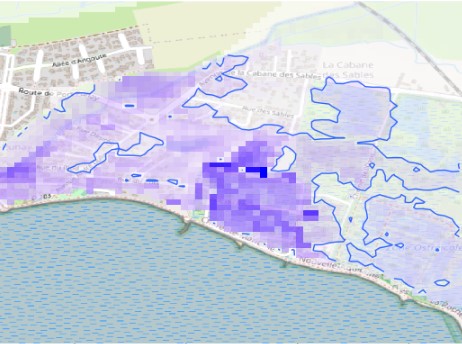}
\caption{$1.4\times 10^{-3}$, 
1 in 714}
\label{clust5}
  \end{subfigure}
  \caption{The four non-empty prototype maps, computed to minimize the expected distance between a random map and its nearest prototype. The coast is at Les Boucholeurs, France, and the maps are $1.6\times1.6 ~\mathrm{km}^2$. The probability mass associated to each prototype map is the probability that a random map is closer to this map than to the other prototypes, for a given distance. The frequence, i.e., the inverse of the probability is added for an easier interpretation. More details about this case study are given in Section~\ref{flooding_sect}. 
}
\label{example_quanti}
\label{quanti_intro}
\end{figure*}

With this visualization, one can picture the most typical flooding scenarios, like Figure~\ref{clust4} where water invades a large area, or more intense and concentrated flooding like Figure~\ref{clust5}.

Lloyd's algorithm is not a new quantization technique (\citealt{MacQueen}), and has evolved since its introduction (\citealt{Bock}). 
This paper presents new extensions to quantization that address the following bottlenecks.
First, the quantization is done here in the context of costly numerical simulations, which makes the computation of the centroids more difficult. Indeed, the flood maps are computed with hydraulic codes that can take several hours to simulate the water levels associated to each set of input parameters.
Plus, flooding is a rare event. Classical Monte Carlo sampling is inefficient for extreme events with a very large return period (typically 1000-10000 years).
Finally, the quantization is performed in a space of pixelated maps ($64\times 64$) which raises metamodeling and storage questions when a very large number $(> 10^6)$ of realizations has to be handled.

The paper is structured as follows. 
Section~\ref{background} defines the quantization problem and the importance sampling estimator to compute the prototype maps. 
Section~\ref{Metamodelisation} introduces the metamodeling scheme built to overcome the duration of the hydraulic simulations. 
Section~\ref{algoandperf} presents the algorithm and the performance metrics that will be evaluated on the two applications described in Section~\ref{campbell_sect} (analytical test case) and Section~\ref{flooding_sect} (real coastal flooding). 
Finally, Section~\ref{perspectives} discusses the main results and gives possible extensions to the method. 
A list of symbols is also provided at the end.

\FloatBarrier

\section{Finding prototype maps} \label{background}

\subsection{Quantization}

Let us consider a random vector $X\in \Dspace$, with $\Dspace$ a subset of $\mathbb{R}^{d},d \in \mathbb{N}$. 
In our context, $X$ is a parametric description of the offshore hydrodynamic (like the tide and surge parameters)
and the dyke breach conditions. 
We define a random function $Y$ from $\Dspace$ the space of offshore conditions and breach parameters to $\Ecal$ the space of possible flood maps, with an inner product $\langle .,.\rangle_{\Ecal}$, and a metric $\norm{.}$ induced by the inner product:
\begin{align*}
  Y \colon \Dspace &\to \Ecal\\
  x &\mapsto Y(x).
\end{align*}
In other terms, with this general formulation, $Y$ associates a distribution of flood maps to every vector of input conditions $x$. 
In practice, $Y$ can be deterministic and thus associate a single flood map to every $x$ which we write $y(x)$.

The quantization problem is to find a set $\Gamm = \{\centro_{1}, \centro_{2},\dots,\centro_{\card}\} \in \Ecal^{\card}$ of $\card \in \mathbb{N}$ representatives of $Y(X)$ in a mathematically relevant way. 
The notation $Y(X)$ should be understood as the map $Y$ conditioned by the random vector $X$ as in regression, $Y(X) \equiv Y\!\!\mid\!\! X$.
The prototype maps will be these representatives in our study. 
More precisely, we first define for $\Gamm \in \Ecal^{\card}$ the closest candidate representative as the function $\rep$:  
\begin{align*}
  \rep \colon \Ecal &\to \Gamm\\
  y &\mapsto \rep(y) = \underset{\centro_{i} \in \Gamm} {\argmin}\norm{y-\centro_{i}} ~.
 \end{align*}
  
A measure of the quantization error is the average distance between a randomly chosen map and its closest representative 
(\citealt{Pages}),
\begin{equation*}
e(\Gamm) = \left[\Esp{\norm{Y(X)-\rep(Y(X))}^{2}}\right]^{\frac{1}{2}} ~.
\end{equation*}
To each centroid $\centro_{j}$ in $\Gamm$, we associate a Voronoi cell $\clust{\Gamm}{j}$
which is the subset of points in $\Ecal$ closest to $\centro_j$,
\begin{equation*}
\clust{\Gamm}{j} = \{y \in \Ecal, \rep(y) = \centro_{j}\} ~.
\end{equation*}

The objective is to minimize this quantization error, i.e., to find 
\begin{align*}
 \Gamm^{\star} = \{\centro^{\star}_{1},\dots, \centro^{\star}_{\card}\} &\in \underset{\Gamm \in \Ecal^{\card}}{\argmin}\: (e(\Gamm)) \\
  &\in \underset{\Gamm \in \Ecal^{\card}}{\argmin}\:\left[\Esp{\norm{Y(X)-\rep(Y(X))}^{2}}\right]^{\frac{1}{2}} \\
 &\in \underset{\Gamm \in \Ecal^{\card}}{\argmin} \: \left[\Esp{\underset{i \in \{1\dots\card\}}{\min} \: \norm{Y(X) - \centro_{i}}^{2}}\right]^{\frac{1}{2}}~.
\end{align*}

Finding the optimal centroids is in general a NP hard problem (\citealt{Mahajan}). 
However, the solution has a strong attribute (\citealt{Pages2}):

\begin{theorem} [Optimal quantization, Kieffer, Cuesta-Albertos] 
\label{pointfixe}
~\\
If 
\begin{minipage}[t]{0.8\textwidth}
\begin{itemize}[label=\raisebox{0.25ex}{\tiny$\bullet$}] 
    \item $\Ecal$ is of finite dimension m ~,
    \item $\forall y,y' \in \Ecal,\scalprodE{y}{y'} = \sum_{i = 1}^{m} \lambda_{i}y_{i}y'_{i}$ with $\forall i, \lambda_{i} > 0$~,
    \item $\Esp{\norm{Y(X)}^{2}} < +\infty$~,
\end{itemize} 
\end{minipage}
\\
then ~\quad $\forall i \in \{1\dots\card\}, \quad \Esp{Y(X)\mid \yinclust{Y(X)}{\Gamm^{\star}}{i}} = \centro^{\star}_{i}$~.
\end{theorem}
This means that the representatives of an optimal quantification coincide with the Voronoi cell centroids.
The proof is given in Appendix \ref{proof_theorem}.
In our context, the distance induced by the scalar product $\scalprodE{y}{y'} = \sum_{i = 1}^{n} \lambda_{i}y_{i}y'_{i}$ between two maps is the mean squared difference between the water depth at each pixel of the maps, weighted by $\lambda_{i}$ which can represent the importance of the zone (for example the housing density, the presence of factories, etc.).

The popular Lloyd's algorithm, also called K-means, is based on Theorem \ref{pointfixe}.
\begin{algorithm}[H]
\begin{algorithmic}[1]
\STATEx {$\Gamm^{[0]} \gets \{\centro^{[0]}_{1},\dots,\centro^{[0]}_{\card}\}$} ~,~$k \gets 0$
\WHILE{stopping criterion not met}
\STATEx {$\centro^{[k+1]}_{j} \gets \Esp{Y(X)\mid \yinclust{Y(X)}{\Gamm^{[k]}}{j}}, j \in \{1,\dots,\card\}.$
}
\STATEx $k \gets k+1$
\ENDWHILE
\end{algorithmic}
\caption{Lloyd's algorithm}
\end{algorithm}
The stopping criterion of the algorithm can be a number of iterations or a lower bound on a distance between $\Gamm^{[k]}$ and $\Gamm^{[k+1]}$. 
The principle of the iterations is illustrated in Figure~\ref{voronoi} in 2D: 
the dots are 5 possible starting points
$\{\centro_{1}^{[0]},\dots, \centro_{5}^{[0]}\}$, 
their associated Voronoi cells are delimited by the black lines, 
the centroids of these cells are computed for a uniform distribution and are drawn as the triangles. 
This Figure shows only one iteration and does not represent the convergence of the algorithm.

\begin{figure}[ht]
\centering
\includegraphics[width=0.8\textwidth]{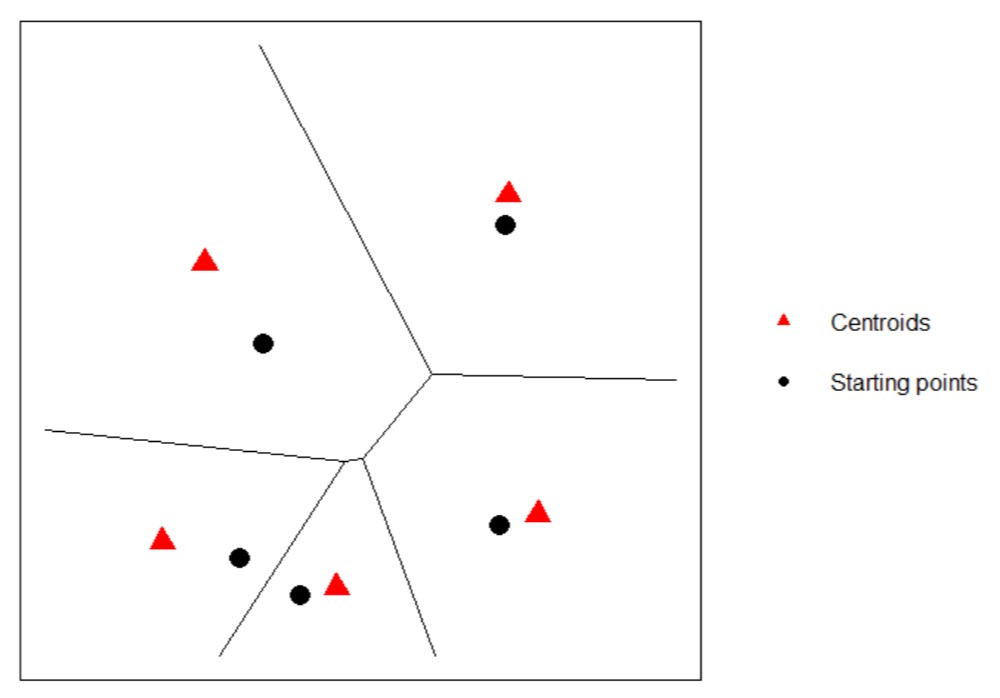}
\caption{An example of Voronoi cells associated to initial points and their centroids in 2D.}

\label{voronoi}
\end{figure}

\subsection{Dealing with low probability events}
 
As stated before, Lloyd's algorithm requires to compute at each iteration and for each Voronoi cell $\clust{\Gamm}{j}, j \in \{1,\dots, \card\}$, $\Esp{Y(X)\mid \yinclust{Y(X)}{\Gamm^{[k]}}{j}}$. The most natural approach would be to sample $Y(X)$ from its density and then to calculate the empirical mean of $Y(X)$ given $\yinclust{Y(X)}{\Gamm^{[k]}}{j}$. 
However, flooding being a rare event, we expect to have one Voronoi cell with almost only empty maps associated to a probability close to 1, and other Voronoi cells with more critical floodings associated to very low probabilities maps. 
Thus, the straightforward sampling is not relevant to locate the Voronoi cells of low probabilities because their member maps will almost surely be missed. 
Importance sampling is thus needed to perform a sampling less concentrated on a single cluster.
  
We will need to estimate centroids and probabilities, which reduces to calculating $\Esp{h(Y(X))}$ with $h \colon \Ecal \mapsto \mathcal{H}$ for some set $\mathcal{H}$. 
To compute the probability that $Y$ is in $A$, $P(Y(X) \in A)$, $h$ is the indicator function $h \colon y \mapsto \mathds{1}_{y \in A}$ with $A \subset \Ecal$ and $\mathcal{H}=\{0,1\}$. 
Alternatively $h$ is $h \colon y \mapsto y\mathds{1}_{y \in A}$ with $\mathcal{H}= \Ecal$ to calculate the mean of $Y$ when it belongs to $A$.

The importance sampling technique is based on the fact that the representation
of \\ $\Esp{h(Y(X))}$ as an expectation is not unique (\citealt{Ecuyer}). 
Indeed, if we consider $\tilde{X}$ a random variable with density function $\biais$ with a support such that $\supp(f_{X}) \subset \supp(\biais) \subset \Dspace$, then
$$\Esp{h(Y(X))} = \Esp{h(Y(\tilde{X}))\frac{f_{X}(\tilde{X})}{\biais(\tilde{X})}}~.$$

The importance sampling estimator of the mean comes from this last representation:
\begin{equation}
\estim_{n}^{IS}(h) = \frac{1}{n} \sum^{n}_{k=1}  h(Y(\tilde{X}^k))\frac{f_{X}(\tilde{X}^k)}{\biais(\tilde{X}^k)} \in \mathbb{R}^{p}~,\label{eq:ISestimator}
\end{equation}  
with $(\tilde{X}^k)^{n}_{k=1}$ i.i.d. random variables of density function $\biais$.
It is well known that such importance sampling estimator is unbiased and consistent (\citealt{tabandeh}) .

The covariance matrix of the estimator is \begin{equation}\label{covariance}
    \var\left(\estim_{n}^{IS}(h)\right) = \frac{1}{n}\var\left(h(Y(\tilde{X}))\frac{f_{X}(\tilde{X})}{\biais(\tilde{X})}\right) ~.
\end{equation}

The idea of the importance sampling method is to choose a relevant density function $\biais$, i.e., a density function that minimizes the terms of the diagonal of the above covariance matrix. 

\subsection{Importance sampling combined with quantization}\label{ISKM}
 
In Lloyd's algorithm, we need to repeatedly compute the expectations of $Y$ within the candidate Voronoi cells. At the $p^{th}$ iteration, these expectations are, 
\begin{equation}
\centro^{[p]}_{j} = \Esp{Y(X)\mid \yinclust{Y(X)}{\Gamm^{[p-1]}}{j}} = \frac{\Esp{Y(X)\mathds{1}_{\yinclust{Y(X)}{\Gamm^{[p-1]}}{j} }}}{\mathbb{P}(\yinclust{Y(X)}{\Gamm^{[p-1]}}{j})}
\label{eq-centroid}
\end{equation}

As for a given $\Gamm$, $\mathbb{P}(\yinclust{Y(X)}{\Gamm}{j}) = \Esp{\mathds{1}_{\yinclust{Y(X)}{\Gamm}{j}}}$ is not known, we rely on the importance sampling estimators for both the numerator and the denominator and we introduce

\begin{equation}\label{estims_is}
  \left\{
    \begin{array}{ll}
    \hat{P}_{n}(\Gamm, j, Y) &= \frac{1}{n} \sum^{n}_{k=1} \mathds{1}_{\yinclust{Y(\tilde{X}^k)}{\Gamm}{j} }\frac{f_{X}(\tilde{X}^k)}{\biais(\tilde{X}_{k})}~,\\
    \estim_{n}(\Gamm,j, Y) &= \frac{\frac{1}{n} \sum^{n}_{k=1} Y(\tilde{X}^k)\mathds{1}_{\yinclust{Y(\tilde{X}^k)}{\Gamm}{j} }\frac{f_{X}(\tilde{X}^k)}{\biais(\tilde{X}_{k})}}{\hat{P}_{n}(\Gamm, j)}~,
\end{array}
\right.
\end{equation}
\noindent the estimators of $\mathbb{P}(\yinclust{Y(X)}{\Gamm}{j})$ and $\Esp{Y(X)\mid \yinclust{Y(X)}{\Gamm}{j}}$ respectively. 
The choice of the distribution $\biais$ is delicate and will be further discussed in Section~\ref{flooding_sect}.

The overall approach was described so far in the very general context of a random field $Y$. In the following, these estimators will be computed with $y$ the deterministic output of the hydraulic simulators, and $\Ypred$ its predictor.

\section{Dealing with time-consuming experiments}\label{Metamodelisation}

In order to have an estimator of $\Esp{\Ytrue(X)\mid \yinclust{\Ytrue(X)}{\Gamm}{j}}$ of sufficiently small variance for all cells $j \in \{1,\dots, \card\}$, we need to compute $\Ytrue(\tilde{X}^k)$ for a sample $(\tilde{X}^k)^{n}_{k=1}$ of large size $n$, with $y$ the deterministic output of the hydraulic simulator.
Realistic physical simulations would be too costly for that purpose, 
and a metamodel is required to approximate $\Ytrue(X)$. 
Gaussian processes (GP) (\citealt{Rasmussen}) are used here to replace some of the simulations.
GPs are readily applicable to low dimensional outputs (most often scalars) but not to $\Ytrue(X)$, the flood map, which is in $\mathbb{R}^{4096}$ ($64\times64$).  
To overcome that, \citealt{Perrin} implemented a method based on functional principal component analysis (FPCA).
Their strategy, which we rely on, is to work with the projections of the maps in a space induced by a basis of $\npc\ll 4096$ maps only. 
The predictions are expressed as coordinates of a map in this basis, through a Gaussian process approach. 
In the following, we call ``metamodel`` the whole procedure: the FPCA followed by a Gaussian process regression.

\subsection{Projection on a basis of maps}

The objective here is to find a basis of $\npc$ maps, with $\npc$ small enough, and to compute the projection of the maps onto this basis,
thus representing every map $y(x)$ as a linear combination of the $\npc$ basis maps: 
\begin{equation*}
    \Yproj(x) \eqdef t_{1}(x)\Ypca_{1} + t_{2}(x)\Ypca_{2} + \dots + t_{\npc}(x)\Ypca_{\npc}~.
\end{equation*}

For this purpose, a functional principal components analysis (FPCA) is carried out on a training database of maps $(\Ytrue(x^{i}))_{i=1,\dots,\ntrain}$. The steps are the following, and further details are given in Appendix \ref{fpca_annex}. : 
\begin{itemize}
    \item Find a decomposition of these maps, seen as realizations of a spatial function, on a functional orthonormal basis $\Phi = (\Phi_{1},\dots,\Phi_{K})^{\top}$ :  $\Ytrue(x^{i}) = \Phi^\top \alpha(x^{i})$ 
    \item Reduce the number of basis functions by keeping the $\tilde{K}$ most important ones, $\tilde{\Phi}$, in terms of their contribution to the energy of the coordinates in the database $(\tilde{\alpha}(x^{i}))_{i=1,\dots,\ntrain}$.
    \item Perform a PCA on the selected coordinates to finally get the principal components $(t(x^{i}))_{i=1,\dots,\ntrain}$ with $t(x^{i}) = (t_1(x^{i}), \dots, t_{\npc}(x^{i}))^{\top}$ and the $\npc\times \tilde{K}$ projection matrix $\Omega$ : $\tilde{\alpha}(x^{i}) \simeq \Omega^{\top} t(x^{i})$.
\end{itemize}

The obtained FPCA-basis maps are $\Omega\tilde{\Phi}$. 
The two basis maps $\Ypca_{1}$ and $\Ypca_{2}$ of the coastal flooding example are shown in Figure~\ref{bases_sr}. 
Negative values can appear on these FPCA maps, but positive depths are recovered with the linear combinations of the principal components.

\begin{figure}[h!]
\centering
\includegraphics[width=1\textwidth]{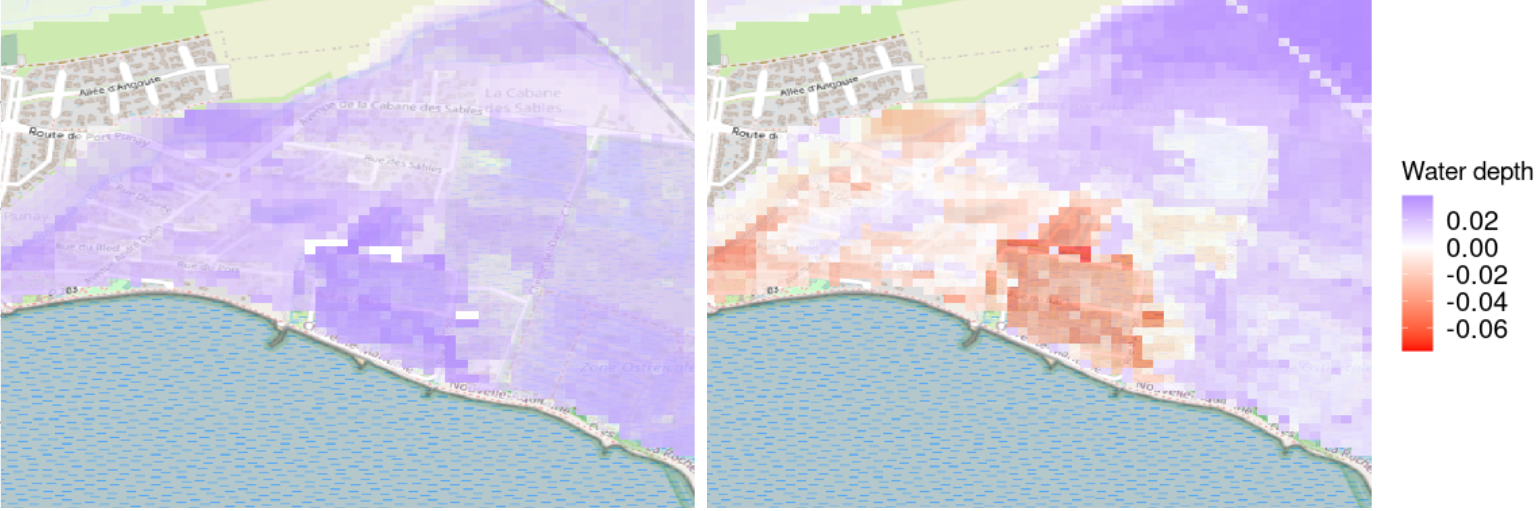}
\caption{$\Ypca_{1}$ and $\Ypca_{2}$, the two FPCA maps of the coastal flooding application without breach.}
\label{bases_sr}
\end{figure}

Figure~\ref{yworst} shows the map $y(x^{\mathrm{worst}})$ the most different from its projection $\Yproj(x^{\mathrm{worst}})$, among the 1300 maps of the training database $\mathcal{B}_{\mathrm{train}}$,
$$y(x^{\mathrm{worst}}) = \underset{y(x) \in \mathcal{B}_{\mathrm{train}}}{\argmax} \norm{y(x)-\Yproj(x)}~.$$
\\
The norm $\norm{~}$ selected here is induced by the usual inner product $\scalprodE{y}{y'} = \sum_{i = 1}^{m} y_{i}y'_{i}$.
\\

One can see that even in the worse case, the map $Y$ and its projection on the FPCA basis $\Yproj$ are quite close. 
This highlights that FPCA indeed captures most of the information.
A map $y(x)$ is satisfactorily compressed into a vector of 2 components, $\left(t_1(x),t_2(x)\right)$.

\begin{figure}[ht]
\centering
\includegraphics[width=1\textwidth]{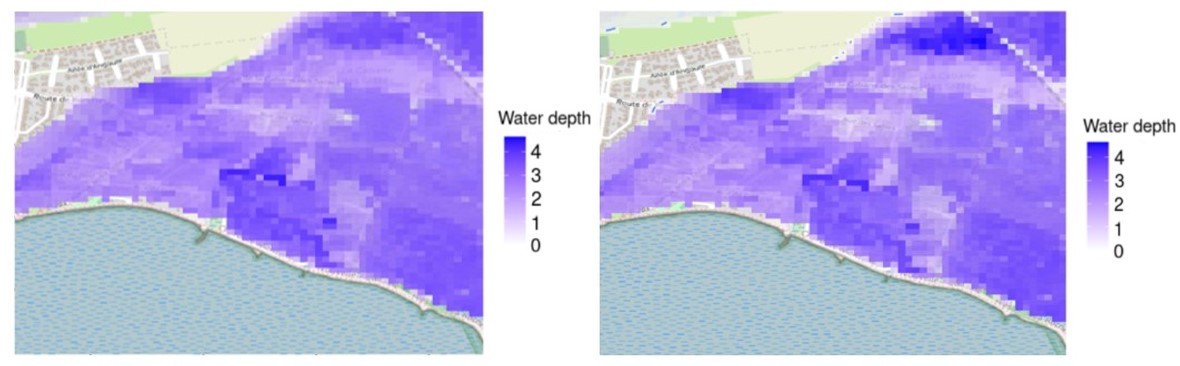}
\caption{Most distant maps before and after FPCA: true $y(x^{\mathrm{worst}})$ and its projection $\Yproj(x^{\mathrm{worst}})$.}
\label{yworst}
\end{figure}

\subsection{Gaussian process regression}

We now aim at predicting the flood map for a new input $x^{\star}$ by predicting its coordinates in the FPCA basis $(t_{1}(x^{\star}), \dots, t_{\npc}(x^{\star}))$. 
For simplicity these coordinates are assumed independent, but the study of their dependency may be relevant in some cases (\citealt{Gu}).
Therefore, $\npc$ one-dimensional Gaussian processes define the metamodel for the flood maps. 
Let $t \colon \Dspace \subset \mathbb{R}^{d} \longrightarrow \mathbb{R}$ be a deterministic function. 
$t$ is approximated as a realization of a Gaussian process $Z(x)$ of prior mean $\mu(x)$ and prior covariance kernel $k(x, x')$. The kernel contains the spatial dependencies between 
$y(x)$ and $y(x')$. 
In our study we use the stationary Matérn 5/2 kernel (\citealt{palar}).

We consider a design of experiments $\mathbb{X}_{\mathrm{train}} = (x^{1},\dots,x^{\ntrain})$ and the observations $Z_{\mathrm{train}} = (z_{i})_{i \in I_{\mathrm{train}}} = (t(x^{i}))_{i \in I_{\mathrm{train}}}$ with $I_{\mathrm{train}} = \{1,\dots, \ntrain\}$. $t$ can be approximated using the distribution of $Z(x) \mid Z(\mathbb{X}_{\mathrm{train}}) = Z_{\mathrm{train}}$, which is a Gaussian process with the following mean (\citealt{Rasmussen}): 
$$
 \hat{t}(x) = \mu(x) + k(x, \mathbb{X}_{\mathrm{train}})k(\mathbb{X}_{\mathrm{train}},\mathbb{X}_{\mathrm{train}})^{-1}(Z_{\mathrm{train}}-\mu(\mathbb{X}_{\mathrm{train}})) $$
where
$k(\mathbb{X}_{\mathrm{train}},\mathbb{X}_{\mathrm{train}}) = [k(x^{i},x^{j})]_{i,j \in I_{\mathrm{train}}}$ and $k(x, \mathbb{X}_{\mathrm{train}}) = [k(x,x^{i})]_{i \in I_{\mathrm{train}}}$. 
This method is applied to all the FPCA components
$i \in \{1,\dots, \npc\}$, for which the GP mean is a predictor $\hat{t}_{i}(x^{\star})$ of $t_{i}(x^{\star})$.
Then, $\hat{y}(x^{\star})$, a predictor of $y(x^{\star})$ is obtained by inverse FPCA as explained in Appendix \ref{inverse_fpca}.
The metamodel errors are investigated in the following sections.

 \section{Algorithm and performance metrics}\label{algoandperf}
 
 Sections \ref{background} and \ref{Metamodelisation} described how to estimate the expectations in Lloyd's algorithm and how to predict a large number of maps that cannot be computed with the time-consuming hydraulic simulators. 
This section details the overall scheme for calculating the prototype maps and their probability masses, and then explains how to evaluate its robustness.

 \subsection{Prototype Maps Algorithm}

The Prototype Maps Algorithm incorporates, inside Lloyd's algorithm, the metamodeling of spatial output and the importance sampling scheme. 

\begin{algorithm}[H]
\begin{algorithmic}[1]
\REQUIRE $(\Ytrue(x^{i}))_{i=1,\dots,\ntrain}$, $f_{X}$, $\biais$, minDistance, $\ell$
\STATE Sample $(\tilde{x}^{k})_{k \in \{1,\dots,\npred\}}$ i.i.d. of density function $\biais$
\STATE Compute $(\Ypred(\tilde{x}^{k}))_{1 \leq k \leq \npred}$ from $(\Ytrue(x^{i}))_{i=1,\dots,\ntrain}$ (GP \& FPCA)
\STATE Compute $(\frac{f_{X}(\tilde{x}^{k})}{\biais(\tilde{x}_{k})})_{1 \leq k \leq \npred}$
\STATE Initialize $\Gamm^{[0]} \gets \{\centro^{[0]}_{0},\dots,\centro^{[0]}_{\card}\} \in \Ecal^{\card}$

\WHILE{$\lVert \Gamm^{[k+1]} - \Gamm^{[k]}\rVert > \text{minDistance}$}
\STATEx{ $\centro^{[k+1]}_{j} \gets \estim_{\npred}(\Gamm^{[k]},j, \Ypred), j = 1,\dots,\card$}
\STATEx $k \gets k+1$
\ENDWHILE
\STATE $\Gammapred = \Gamm^{[k]}$
\STATE Compute $\hat{P}_{\tilde{n}}(\Gammapred,j, \Ypred), j = 1,\dots,\card$ 
\end{algorithmic}
\textbf{Output:} $\Gammapred$ and $\hat{P}_{\tilde{n}}(\Gammapred,j, \Ypred)$ 
\caption{Prototype Maps Algorithm}
\label{alg:ProtoMapsAlgo}
\end{algorithm}
The expression for $\hat{P}_{\tilde{n}}(\Gammapred,j, \Ypred)$ and $\estim_{\npred}(\Gamm^{[k]},j, \Ypred)$ stems from applying Equations~\eqref{estims_is} to the Gaussian Process prediction of the maps.
The calculation of $\hat{P}_{\npred}(\Gammapred,j,\Ypred)$ at step 8 could be extracted from the last iteration since it is the denominator of $\estim_{\ntrue}(\Gamm^{\star},j,\Ypred)$. 
However, these quantities are the estimators of the probabilities that a flood map is part of each Voronoi cell, for which we want to be as precise as possible, we thus re-compute $\hat{P}_{\tilde{n}}(\Gammapred,j,\Ypred)$ at the end with $\tilde{n} > \npred$.

The overall complexity of the Prototype Maps Algorithm is $\mathcal{O}((\tilde{n}+\npred)\ntrain^2\npc + \card s^2(\npred n_{\mathrm{it}} + \tilde{n}))$ with $n_{\mathrm{it}}$ the number of iterations. The calculation of the complexity is detailed in \citealt{sire_complex}.
The algorithm requires the storage of $(\Ypred(\tilde{X}^k))_{1 \leq k \leq \npred}$ of size $\npred s^2$ (with $s = 64$ here), which is a limiting factor.

If the map simulations are not expensive, the algorithm can of course be performed directly with the true maps,  $(\Ytrue(x^{i}))_{i=1,\dots,\ntrue}$, instead of the GP predictions.
minDistance $= 10^{-16}$ is chosen here in the applications.

\subsection{Performance metrics}

We expect that Lloyd's algorithm returns $\card$ prototype maps associated to the minimum quantization error, as well as $\card$ probabilities of maps to belong to each of the induced Voronoi cell. The errors will be evaluated both on the centroids and on each cluster probability. 
To ensure the robustness of the metrics,  we consider $(\Gamm^{r})_{r\in \{1,\dots,n_{\Gamm}\}}$ a family of $n_{\Gamm}$ different random quantizations chosen as random perturbations of a priori prototypes. 
Details about the random quantizations are provided in \citealt{sire_is}.  
We also introduce $\Gammaref$ the optimal quantization obtained with the real maps coming from the  simulators, and $\Gammapred$ the optimal quantization obtained with the predicted maps.

\paragraph{Metamodel performance metrics.} 
\label{perf_metrics_metamodel}

The following performance metrics quantify the errors caused by the metamodel (MM) on the quantization and on the cluster probability: 
\begin{itemize}
    \item The excess in quantization error due to the metamodel,
    \begin{equation}\label{error_mm_gamma}\epsilon_{\Gamma}^{MM} = \frac{\hat{e}(\Gammapred)-\hat{e}(\Gammaref)}{\hat{e}(\Gammaref)}~,
    \end{equation}where $\hat{e}(\Gamm) = \left(\frac{1}{\nerr} \sum_{k = 1}^{\nerr} \norm{\Ytrue(\tilde{X}^k) - \rep(\Ytrue(\tilde{X}^k))}^2\frac{f_{X}(\tilde{X}^k)}{\biais(\tilde{X}^{k})}\right)^{\frac{1}{2}}$ is an empirical quantization error.
We take $\nerr = 10^6$ here. 
This error can only be computed in the analytical case (see Section \ref{quant_camp}) as $\Gammaref$ can not be computed in the flooding case whose computational cost is too high.

    \item The relative probability error,
\begin{equation} 
  \erreurprobarelative(n,\Gamm^{r},j) = \frac{\mid  \hat{P}_{n}(\Gamm^{r},j,\Ytrue) - \hat{P}_{n}(\Gamm^{r},j,\Ypred) \mid}{\hat{P}_{n}(\Gamm^{r},j,\Ytrue)}~. 
\label{error_mm_proba}
  \end{equation}

computed for all $r = 1,\dots,n_{\Gamm}$ and $j = 1,\dots,\card$. 
Notice that the prototypes are identical in both probabilities, and only the effect of the metamodel on the cluster probability is assessed.
\end{itemize}

\paragraph{Importance sampling performance metrics.} 
\label{perf_metrics_is}

The effect of the importance sampling (IS) on the quantization is quantified, for a given Voronoi cell $\clust{\Gamm^r}{j}$, through two metrics:
\begin{itemize}
    \item The IS coefficient of variation of the membership probability, $\epsilon_P^{IS}(\tilde{n},\Gamm^r,j)$, which is an estimator of $\frac{\sqrt{\var\left(\hat{P}_{\tilde{n}}(\Gamm^r,j, \Ytrue)\right)}}{\mathbb{E}\left(\hat{P}_{\tilde{n}}(\Gamm^r,j, \Ytrue)\right)}$.
It measures an error due to the importance sampling in the estimation of the probability
of the Voronoi cell number $j$. .
    \item The IS centroid standard deviation, $\epsilon_{\Gamm}^{IS}(\ntrue,\Gamm^{r},j)$,
which is the $90\%$-quantile of the diagonal terms of the covariance matrix $\var\left(\estim_{\ntrue}(\Gamm^{r},j,\Ytrue)\right)$. 
It characterizes the error coming from importance sampling in the estimation of the centroid of the Voronoi cell number $j$. 
Among other things, it will show potential variations due to the denominator being close to zero in the centroid formula~\eqref{eq-centroid}.
\end{itemize}

These quantities are calculated for $r = 1,\dots,n_{\Gamm}$ and $j = 1,\dots,\card$ and their distributions serve as the metrics for the importance sampling performance. Details about the computation of these errors are given in \citealt{sire_is}.
 
 \section{Application to an analytical case}\label{campbell_sect}

\subsection{Campbell test case}

The performance of the method is first evaluated on an analytical test case, called the Campbell2D function (\citealt{Marrel}). 
Contrarily to realistic flooding simulators, a very high number of calls to the Campbell function are possible, 
which makes it an appropriate case for testing the Prototype Maps Algorithm.
This function has eight inputs ($d=8$) and a spatial map as output, i.e., its output is a function which depends on two inputs $(z = (z_{1},z_{2}))$ corresponding to spatial coordinates, 

\begin{align*}
\campbell\colon \;\;\;\;\;\;\;\;\;\; [-1,5]^{8} &\longrightarrow \mathbb{L}^{2}([-90,90]^{2}) \\
x = (x_{1},\dots, x_{8}) &\longmapsto \campbell_{x}(z)~,
\end{align*}
where $z = (z_{1}, z_{2}) \in [-90,90]^{2}$, $x_{j} \in [-1,5]$ for $j = 1,\dots,8$ and 
\begin{align*}
\campbell_{x}(z_{1}, z_{2}) &=    x_{1}\exp\left[-\frac{(0.8z_{1}+0.2z_{2}-10x_{2})^2}{60x_{1}^2}\right] \\
&+  (x_{2}+x_{4})\exp\left[\frac{(0.5z_1+0.5z_2)x_{1}}{500}\right] +  x_{5}(x_{3}-2)\exp\left[-\frac{(0.4z_{1}+0.6z_{2}-20x_{6})^2}{40x_{5}^2}\right] \\
&+ 
(x_{6}+x_{8})\exp\left[\frac{(0.3z_1+0.7z_2)x_{7}}{250}\right]
\end{align*}

\noindent Examples of Campbell maps are shown in Figure~\ref{examples_campbell}.

\FloatBarrier

\begin{figure*}[t!]
\centering
  \begin{subfigure}[t]{0.48\textwidth}
\includegraphics[width=1\linewidth]{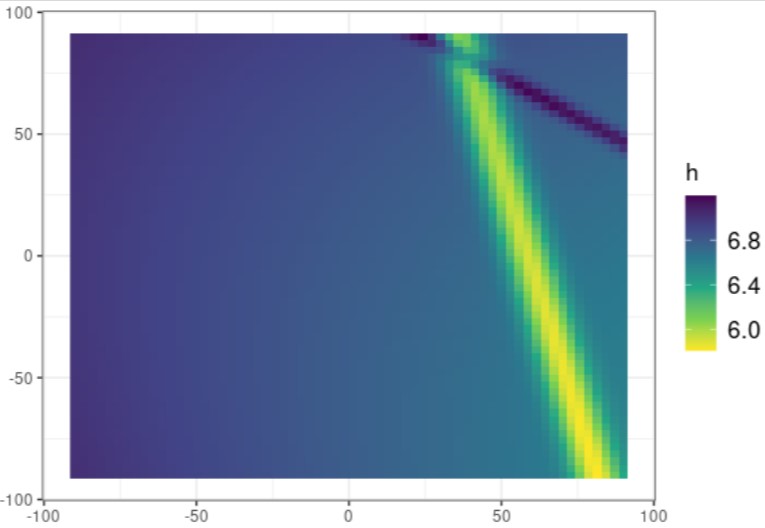}
\caption{}
 \label{example_campbell_1}
 \end{subfigure}%
 ~ 
\begin{subfigure}[t]{0.48\textwidth}
\includegraphics[width=1\linewidth]{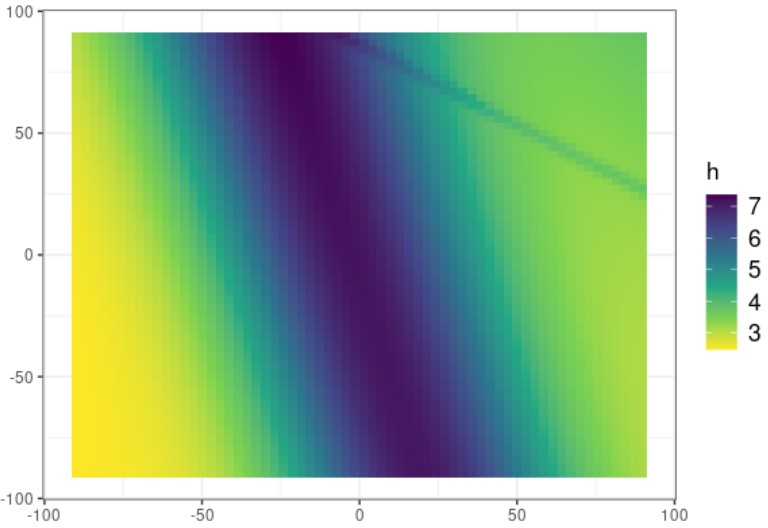}
\caption{}
\label{example_campbell_2}
  \end{subfigure}
  \caption{Two examples of Campbell maps. \ref{example_campbell_1} for $x = (-0.8,4.7,1,-0.1,-0.4,3.2,0.7,-1)$ and  \ref{example_campbell_2} for $x = (4.5,-0.3,3.5,1.4,0.3,2.6,-0.6,-1)$}
\label{examples_campbell}
\end{figure*}

\FloatBarrier

As in the forthcoming flooding case, the output space is $\Ecal = \mathcal{M}_{\dime, \dime}(\mathbb{R})$.
$f_{X}$, the density function of $X$, is the same as the density function of the input variables of the flooding application.
In order to have the same input domain as in the forthcoming flooding example, we impose $x_{8} = -1$ and we introduce an affine transformation, denoted $a$, from $\supp(f_{X})$ to $[-1,5]^{7}$, so that the input space is $\Dspace = \supp(f_{X}) \subset \mathbb{R}^7$. Details are given in Appendix \ref{probas_inputs}.
In the end, the true map function is

\begin{align*}
\Ytrue \colon \Dspace &\to \Ecal\\
x &\mapsto (\campbell_{(a(x)_{1},\dots, a(x)_{7}, -1)}(\mathbf{z}))_{z \in \mathcal{Z}} ~, 
\end{align*}

with $\mathcal{Z} = \{(-90 + \frac{i}{\dime}180,-90 + \frac{j}{\dime}180)\}_{i,j \in \{1,\dots,\dime\}^{2}}$.

The sampling density, $\biais$, is the same as in the flooding case and is described in Section~\ref{sec:flood_description}. The number of prototypes $\ell = 5$ is chosen a priori here but its selection is discussed in Section~\ref{perspectives}.

\subsection{Quantizing the Campbell function}\label{quant_camp}
\paragraph{Campbell metamodel errors.}

As explained is Section~\ref{perf_metrics_metamodel}, we gauge the performance of the empirical quantization with the two following metrics. 

First, the excess in quantization error due the metamodel, $\epsilon_{\Gamma}^{MM}$ 
(Equation \eqref{error_mm_gamma}), is equal to $3.6\%$ when evaluated from $\npred = 10^6$ maps.

Regarding the error on the probability estimators, we evaluate $\erreurprobarelative(10^5,\Gamm^{r},j)$ for all the cells $j = 1,\dots,\ell$ of all the tested quantization $\Gamm^{r}$, $r = 1,\dots,n_{\Gamm}$. Figure~\ref{proba_error_metamodel} illustrates the distribution of these errors. The results are really satisfying: the median and the 75th percentile are below $1\%$. 
The tests on the Campbell function thus show that the metamodel is sufficiently precise in terms of both the prototype maps and the computed probability mass of the Voronoi cells.

\begin{figure*}[t!]
\centering
  \begin{subfigure}[t]{0.3\textwidth}
\includegraphics[width=1\linewidth]{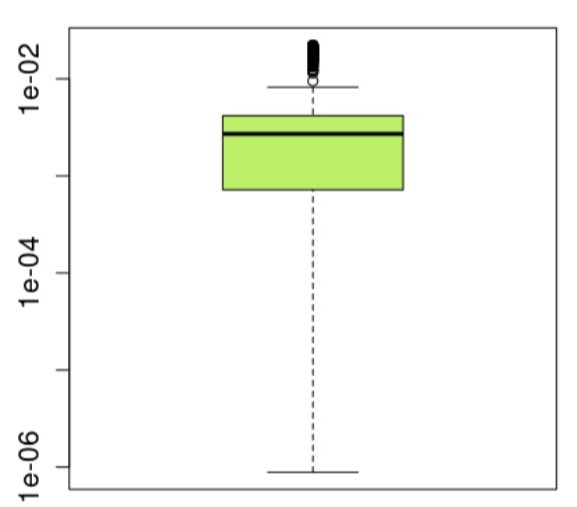}
\caption{Relative probability error}
 \label{proba_error_metamodel}
 \end{subfigure}%
 ~ 
\begin{subfigure}[t]{0.3\textwidth}
\includegraphics[width=1\linewidth]{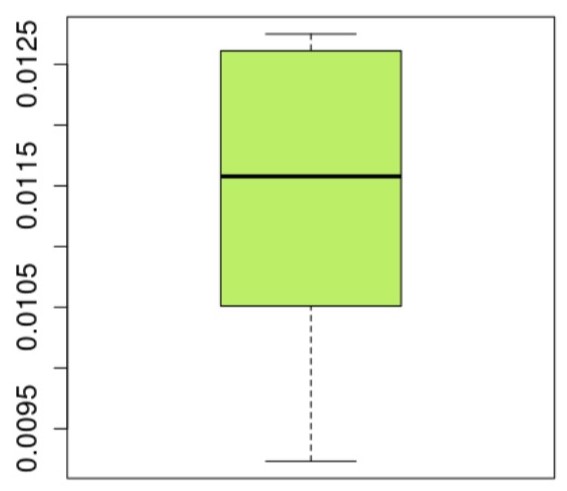}
\caption{IS relative std. deviation}
\label{proba_std_campbell}
  \end{subfigure}
 ~ 
\begin{subfigure}[t]{0.3\textwidth}
\includegraphics[width=1\linewidth]{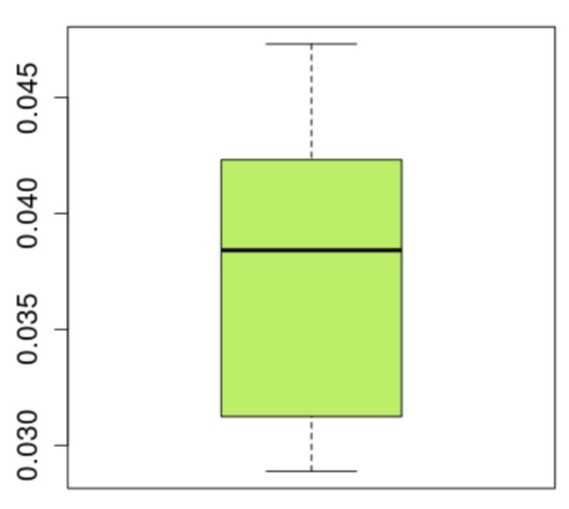}
\caption{IS centroid std. deviation}
\label{centroid_std_campbell}
  \end{subfigure}
  \caption{ 
Boxplots of the distributions of the various errors in the Campbell test case. The distributions correspond to all $r=1,\ldots,100\text{ and } j=1,\ldots,5.$
(a) $\erreurprobarelative(10^5,\Gamm^{r},j)$,  the relative probability error. 
(b) $\epsilon_P^{IS}(10^7,\Gamm^{r},j)$, the IS coefficient of variation of the membership probability. 
(c) $\epsilon_{\Gamm}^{IS}(10^6,\Gamm^{r},j)$, the IS centroid standard deviation.
}
\label{}
\end{figure*}

\paragraph{Importance sampling errors.}

As stated in Section~\ref{perf_metrics_is}, 
the errors related to the importance sampling scheme are measured through the IS coefficient of variation of the membership probability, $\epsilon_P^{IS}(\tilde{n},\Gamm^{r},j)$, and
the IS centroid standard deviation, $\epsilon_{\Gamm}^{IS}(\ntrue,\Gamm^{r},j)$, for all choices of prototypes $r = 1,\dots, n_{\Gamm}$ and 
all clusters $j = 1,\dots,\card$.

The distributions of these metrics are presented in Figures \ref{proba_std_campbell} and \ref{centroid_std_campbell} for $\tilde{n} = 10^7$
and $\ntrue = 10^6$. 
$\epsilon_P^{IS}(\tilde{n},\Gamm^{r},j)$ is around $1\%$, which is of the same order of magnitude as the error from the metamodel, and which is sufficiently accurate. 
$\epsilon_{\Gamm}^{IS}(\ntrue,\Gamm^{r},j)$ is around 0.04, which is also very precise.
These results validate the precision of the importance sampling estimators.

\paragraph{Campbell prototype maps.}

The outputs of the adapted Lloyd's algorithm are the $\card = 5$ representative maps $\Gammaref =\{\centro^{\star}_{1},\dots,\centro^{\star}_{5}\}$ drawn in Figure~\ref{centroids_campbell} which were obtained by using the true maps. 
The probability mass of each Voronoi cell is given below each map with two estimates, one calculated with the true maps and the other with the metamodeled maps. 
The closeness of the results confirms the precision of the metamodel-based approach. 
The same broad backslash-like pattern runs across the 5 maps, which differ by their $h$ levels. 
It is explained by the density function $f_{X}$ illustrated in Appendix \ref{probas_inputs}, that concentrates $x_{2},x_{4}$ and $x_{5}$ in specific zones that boost the occurence of this pattern.

\begin{figure}
\centering
\begin{subfigure}{.32\textwidth}
  \centering
  \includegraphics[width=.8\linewidth]{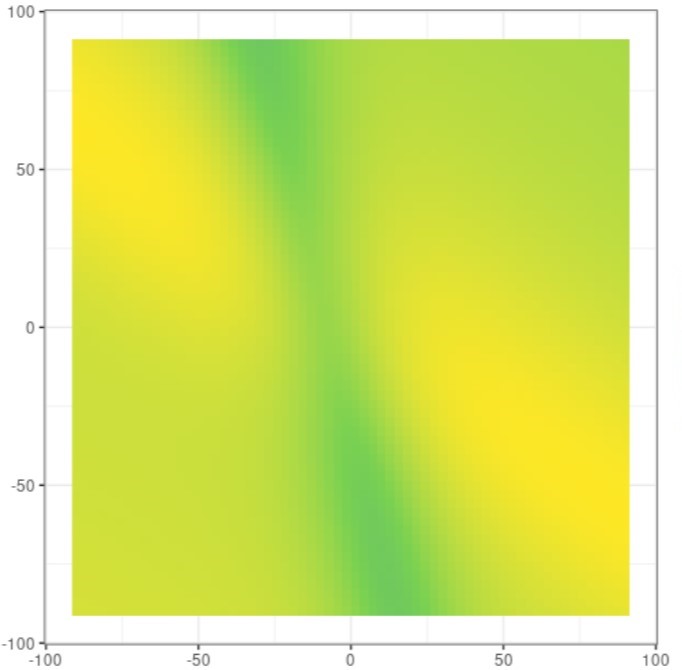}  
  \caption{\scriptsize\color{white}{aa}\color{black}$\hat{P}_{\tilde{n}}(\Gammaref,1,\Ytrue) = 6.6\times10^{-2}\\$
  \color{white}{aaaaaa}\color{black}$\hat{P}_{\tilde{n}}(\Gammaref,1,\Ypred) = 6.2\times10^{-2}$
  }
\end{subfigure}
\begin{subfigure}{.32\textwidth}
  \centering
  \includegraphics[width=.8\linewidth]{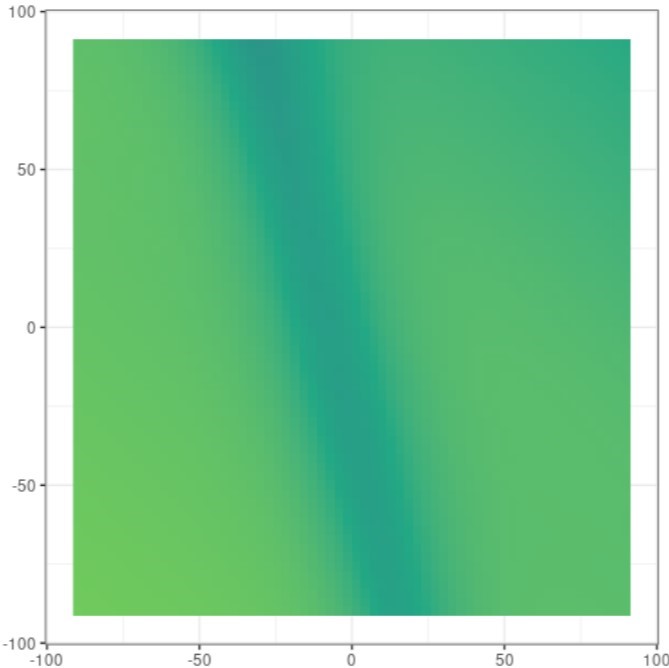}  
  \caption{\scriptsize\color{white}{aa}\color{black}$\hat{P}_{\tilde{n}}(\Gammaref,1,\Ytrue) = 2.2\times10^{-1}\\$
  \color{white}{aaaaaa}\color{black}$\hat{P}_{\tilde{n}}(\Gammaref,1,\Ypred) = 2.2\times10^{-1}$
  }
\end{subfigure}
\begin{subfigure}{.32\textwidth}
  \centering
  \includegraphics[width=.8\linewidth]{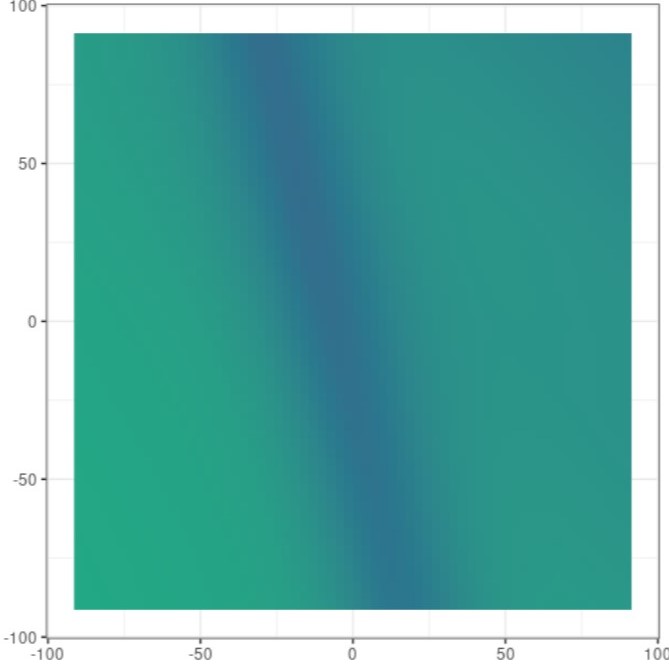}  
\caption{\scriptsize\color{white}{aa}\color{black}$\hat{P}_{\tilde{n}}(\Gammaref,1,\Ytrue) = 2.7\times10^{-1}\\$
  \color{white}{aaaaaa}\color{black}$\hat{P}_{\tilde{n}}(\Gammaref,1,\Ypred) = 2.7\times10^{-1}$
  }
\end{subfigure}

\centering
\begin{subfigure}{.32\textwidth}
  \centering
  \includegraphics[width=.8\linewidth]{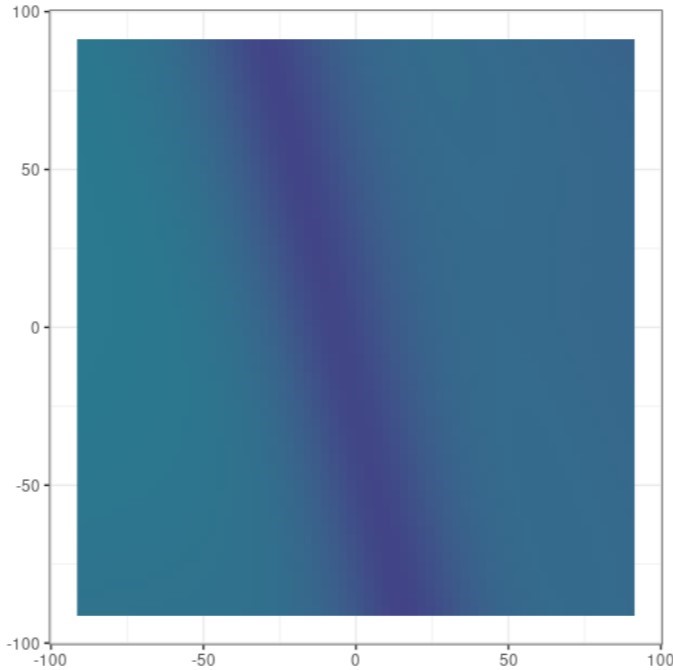}  
    \caption{\scriptsize\color{white}{aa}\color{black}$\hat{P}_{\tilde{n}}(\Gammaref,1,\Ytrue) = 2.6\times10^{-1}\\$
  \color{white}{aaaaaa}\color{black}$\hat{P}_{\tilde{n}}(\Gammaref,1,\Ypred) = 2.6\times10^{-1}$
  }
\end{subfigure}
\begin{subfigure}{.41\textwidth}
  \centering
  \includegraphics[width=.8\linewidth]{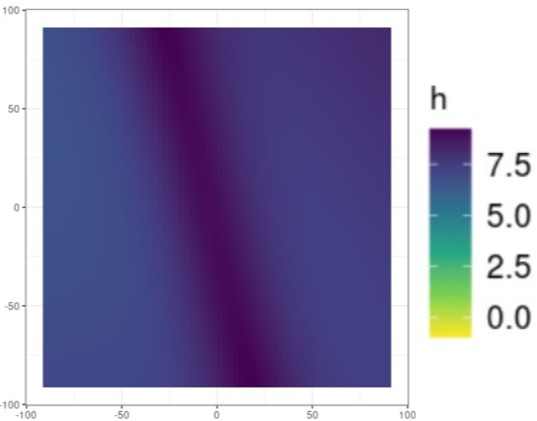}  
    \caption{\scriptsize\color{white}{aa}\color{black}$\hat{P}_{\tilde{n}}(\Gammaref,1,\Ytrue) = 2.0\times10^{-1}\\$
  \color{white}{aaaaaa}\color{black}$\hat{P}_{\tilde{n}}(\Gammaref,1,\Ypred) = 2.0\times10^{-1}$
  }
\end{subfigure}
\caption{The 5 prototype maps obtained from $\ntrue = 10^6$
true maps (without metamodel) in the Campbell case, and their probability masses estimated with $\tilde{n} = 10^7$ true maps (top of subcaptions) and estimated with the predicted maps (bottom of the subcaptions).} 
\label{centroids_campbell}
\end{figure}

\section{Application to flooding}\label{flooding_sect} 

\subsection{Description of the flooding test case}
\label{sec:flood_description}

\paragraph{Parameterization.}
The main application of this article is about the Boucholeurs district located 
on the French Atlantic Coast (Figure~\ref{map_boucholeurs}) near the city of La Rochelle. 
This site was hit by the Xynthia storm in February 2010, which caused the inundation of several areas and severe human and economic damage (\citealt{xynthia}). 

The main coastal flooding processes correspond to overflow and are simulated with the numerical code MARS (\citealt{Lazure}), where adaptations were made by the BRGM to take into account specificities of local water circulation flows around connections (nozzles, spillways, etc. and breaching phenomena). 
We focus on the interplay between tide and storm surge on the spatial distribution of the maximum water depth after flooding in the urban area.
As shown in the Figure~\ref{tidesurge}, the tide is simplified and assumed to be represented by a sinusoidal signal, parameterized by the high-tide level $T$. 
The surge signal is modeled as a triangular function defined by four parameters: the peak amplitude $S$, the phase difference between the surge peak and the high tide $t_{0}$, the time duration of the rising part $t_{-}$ and the time duration of the falling part $t_{+}$. 
In total, 5 variables describe the offshore conditions, $X=(T, S, t_0,t_-,t_+)$.

Each of these input variables is assumed to be independent and described by an empirically defined probability distribution based on the analysis performed by \citealt{Rohmer} of tide gauge measurements from Ile d'Aix \footnote{data available on \url{data.shom.fr}} (the tide gauge closest to the study site) and La Rochelle-La Pallice. 
A constant surge level offset (chosen at 65~cm i.e., compatible with future sea level rise in this region \footnote{see e.g. \url{https://sealevel.nasa.gov/data_tools/17})}) is added to its empirical distribution to ensure that a sufficient number of flooding events are induced. Appendix \ref{probas_inputs} shows the obtained probability distributions. 

Regarding the variables representing a possible breach, two characteristics are taken into account.
First, the location of the breaches. Ten different locations are selected. Six of them (B1-B6, Figure~\ref{map_boucholeurs}) are related to natural dunes and are selected based on observations made during the Xynthia event (\citealt{Muller}). Four other locations (B7-B10) are related to artificial dykes and are added because they are near vulnerable zones.
The other breach-related variable is the topographic level after failure, parametrized as the fraction $p$ (erosion rate) of the initial crest level left.

The distribution of the breach characteristics (location and erosion rate) depends on the offshore conditions. 
If the maximum of the sum of the signals (tide and surge) over time is higher than $70\%$ of the embankment's height, the probability of breaching is $0.5$. 
Vice versa, the probability of breaching is assumed to be $10^{-4}$. 
Finally, we consider that if there is a breach, then the location and the erosion rate follow a uniform distribution $\mathcal{U}_{\{1,\dots,10\}}$ and $\mathcal{U}_{]0,1]}$, respectively. 
This breach model is a compromise between simplicity and realism.

\begin{figure*}[t!]
\centering
  \begin{subfigure}[t]{0.62\textwidth}
\includegraphics[width=1\linewidth]{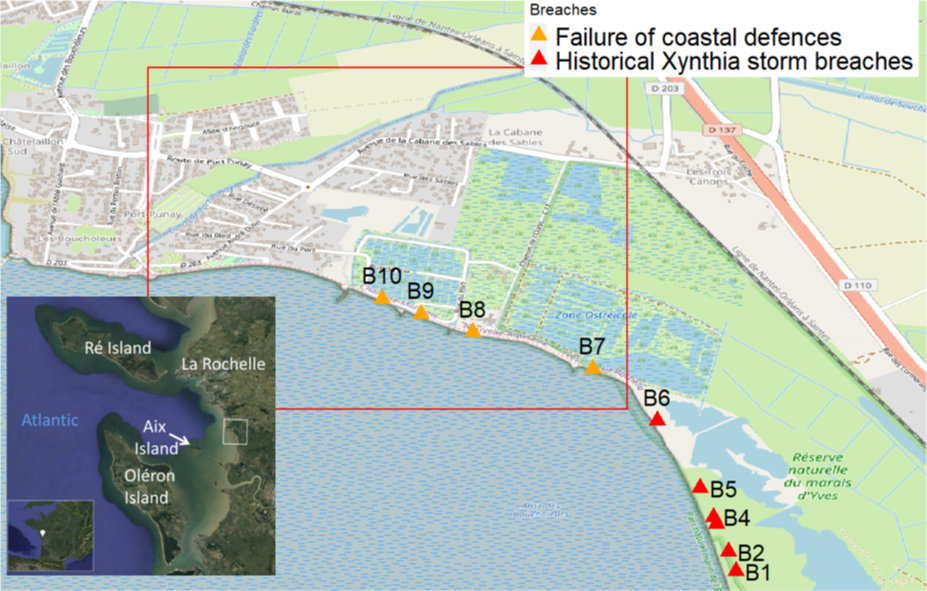}
\caption{Study site of Les Boucholeurs on the Atlantic coast in France, with the locations of the breaches (triangle symbols) (\citealt{Rohmer}). The red rectangle represents the investigated zone where quantization is performed.}
 \label{map_boucholeurs}
 \end{subfigure}
 ~ 
\begin{subfigure}[t]{0.34\textwidth}
\includegraphics[width=1\linewidth]{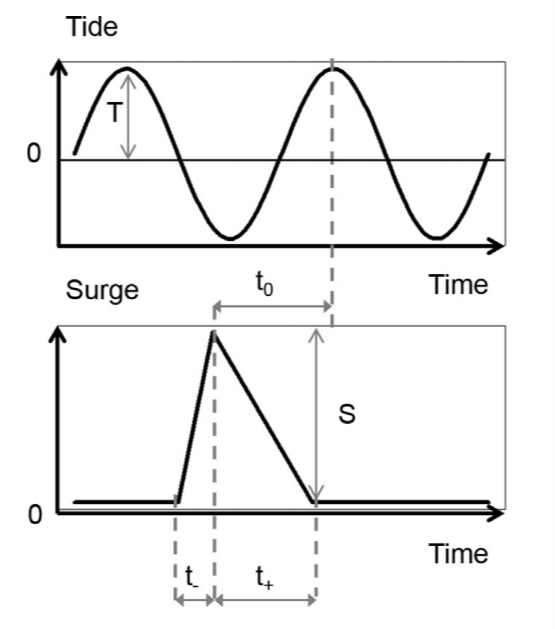}
\caption{Parameterization of the tide and surge temporal signals (see \citealt{Rohmer}).}
\label{tidesurge}
  \end{subfigure}
  \caption{}
\end{figure*}

\paragraph{Data set generation.}
Our study relies on 1300 simulations of flood maps generated as follows: 
\begin{itemize}
    \item The offshore conditions are sampled as the beginning of the Sobol Sequence in 5 dimensions (\citealt{Cheng}).
    \item 500 of the 1300 simulations are run without breach, and the breach parameters of the 800 others are sampled uniformly 
according to $\mathcal{U}_{\{1,\dots,10\}}$ for a single breach location and $\mathcal{U}_{[0,1]}$ for the associated erosion rate.
\end{itemize}

These 1300 maps form the training database, and each map belongs to the space $\Ecal = \mathcal{M}_{\dime, \dime}(\mathbb{R})$. 
The inner product $\langle .,.\rangle_{\Ecal}$ of Theorem \ref{pointfixe} is such that all $\lambda_i$'s are set to 1 so that $\forall (y,y') \in \Ecal^{2}, \langle y,y'\rangle_{\Ecal} = \sum_{i=1}^{\dime}\sum_{j=1}^{\dime} y_{i,j}y'_{i,j}$.  

\paragraph{IS sampling density.}
The biased density function is $\biais = \biais_{1}\times\dots\times \biais_{7}$ with $\biais_{i} = \frac{\mathds{1}_{I_{i}}}{\int_{I_{i}}dx}, i = 1,\dots,5$, where $I_i$ is the support of the marginal density function of the $i^{\mathrm{th}}$ input variable (cf. Appendix~\ref{probas_inputs}),
$\biais_{6} = \frac{1}{10}\sum_{k=1}^{10}\delta_{k}$ and $\biais_{7} = \frac{5}{13}  \delta_{0} + (\frac{8}{13})\mathds{1}_{]0,1]}$.\\
The idea is to choose a single density for all the estimators related to all the Voronoi cells. 

The optimal biased density for which the IS estimators have zero variance is theoretically known (\citealt{Ecuyer}) to be a normalized version of 
$\mathds{1}_{\yinclust{\Ytrue(\tilde{x})}{\Gamm}{j}} f_{X}(\tilde{x})$ or $y(\tilde{x})\mathds{1}_{\yinclust{\Ytrue(\tilde{x})}{\Gamm}{j}} f_{X}(\tilde{x})$.
It changes with each cell $j$, affecting a non-zero probability mass only to the cell. 
It is not practical as it requires the knowledge of the result of the estimation (probability or mean map) for the normalization. 
The choice of a uniform distribution for the variables related to the offshore conditions is an approximation to these optimal densities that doesn't depend on the cell. The uniform density is relevant for quantization as it provides maps covering the whole domain irrespective of their occurrence probability.

\paragraph{Dealing with the concentration of empty maps.}
One of the main difference between the Campbell function and the flooding case is a Dirac mass in the output distribution. 
Indeed, the probability that the map is free of water ($0$ at every pixel of the matrix) is close to 1. 
As Gaussian processes are not adapted to predict empty maps, their direct use in the Prototype Maps Algorithm would overestimate the flooding probability.
As a counter-measure, a random forest  classification step is added (\citealt{Ayyadevara}), creating a hurdle model (\citealt{McDowell}): 

\begin{itemize}
    \item A first step predicts if the water level is below a threshold $\thresh$ (it is a parameter to tune) or not with the classification method.
    \item If the water level is forecasted to be above $\thresh$, 
then use the spatial metamodel described in Section~\ref{Metamodelisation}. 
Otherwise, the predicted map is empty.
\end{itemize}
A final post-processing step is to convert all negative water depth values into null values. 

Another difference with the Campbell case is the existence of two behaviours with or without breach. 
Two metamodels are built independently for each of these two situations.

\paragraph{Initializing Lloyd.}

Another implementation aspect is the choice of the initial prototype maps at the beginning of Lloyd's algorithm (Step 4, Algorithm~\ref{alg:ProtoMapsAlgo}). 
Our a priori is that the future Voronoi cells will be associated with diverse flooding levels. 
The Prototype Maps Algorithm is first run with only the 1300 realistic maps instead of the $\ntrue=10^6$ metamodeled maps. But this degraded version of the algorithm, which provides the initial centroids, also needs to be initialized, 
which is done with the maps associated to the $\card$ equally spaced quantiles of the training water volumes.

\subsection{Quantizing floodings}

\paragraph{Flooding metamodel error.}

As presented in Section~\ref{perf_metrics_metamodel}, to evaluate the impact of the metamodel on the precision of the probability estimators, we compute the relative probability error, $\erreurprobarelative(1300,\Gamm^{r},j)$. 
Note that, unlike the analytical Campbell function, the excess in quantization error, $\epsilon_{\Gamma}^{MM}$, cannot be calculated since the true centroids are not known.

The calculation of the relative probability error is composed of the following steps:

\begin{itemize}
    \item The flood predictions,
$(\Ypred(\tilde{X}^{(j)}))_{1 \leq j \leq 1300}$, are computed from the  $(\Ytrue(\tilde{X}^{(j)}))_{1 \leq j \leq 1300}$ 
 hydraulic simulations within a 10-fold cross-validation procedure. 
Each of the fold involves a classification and a metamodeling from 1170 maps followed by a prediction of 130 maps.
    \item The probabilities  $\hat{P}_{1300}(\Gamm,j,\Ytrue)$ and $\hat{P}_{1300}(\Gamm,j,\Ypred)$, $j \in \{1,\dots,5\}$, are computed 
for the perturbed centroids (generated as explained in Section~\ref{perf_metrics_metamodel}) $\Gamm \in (\Gamm^{r})_{r\in \{1,\dots,n_{\Gamm}\}}$.
\item The relative probability error  $\erreurprobarelative(1300,\Gamm,j)$ is computed following Equation~\eqref{error_mm_proba} for $\Gamm \in (\Gamm^{r})_{r\in \{1,\dots,n_{\Gamm}\}}$ and $j = 1,\dots,5$.
\end{itemize}

The distribution of $\erreurprobarelative(1300,\Gamm^{r},j)$, the relative errors, is shown in Figure~\ref{proba_error_metamodel_flooding} for $n_{\Gamm} = \nbgamma$. 
The median of the error is around $2.5\%$. 
It reveals that the FPCA and the Gaussian process together are effective in predicting the Voronoi cell in which a flood map is located, i.e., in predicting the type of flooding of a new input vector.  This performance is obtained with only two FPCA basis maps that explain $98\%$ of the variance, see Figure \ref{bases_sr}, which highlights the low intrinsic dimension of the phenomenon.

\begin{figure*}[t!]
\centering
  \begin{subfigure}[t]{0.3\textwidth}
\includegraphics[width=1\linewidth]{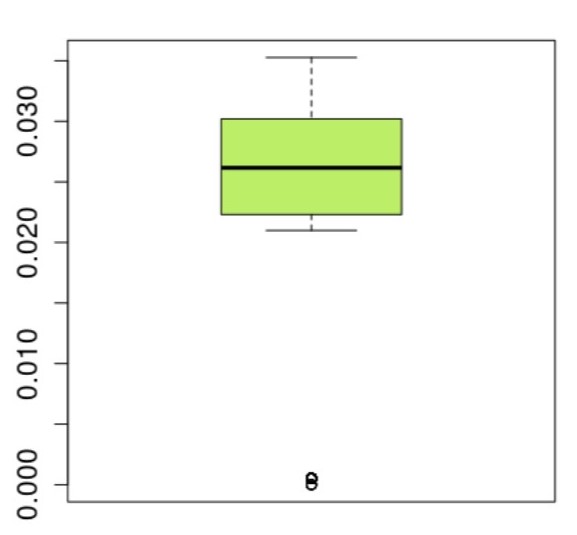}
\caption{Relative probability error}
 \label{proba_error_metamodel_flooding}
 \end{subfigure}%
 ~ 
\begin{subfigure}[t]{0.3\textwidth}
\includegraphics[width=1\linewidth]{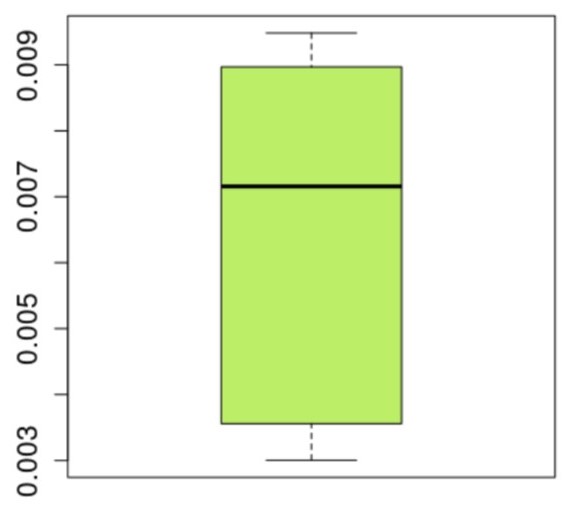}
\caption{IS relative std. deviation}
\label{std_flooding}
  \end{subfigure}
  ~ 
\begin{subfigure}[t]{0.312\textwidth}
\includegraphics[width=1\linewidth]{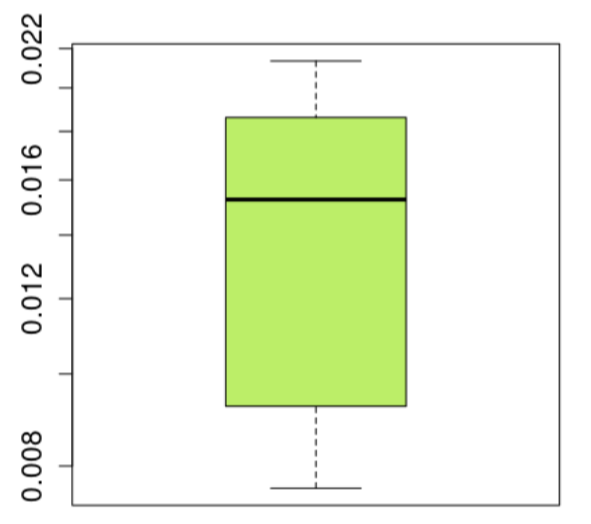}
\caption{IS centroid std. deviation}
\label{centroid_flooding}
  \end{subfigure}
  \caption{ 
Boxplots of the distributions of the various errors in the flooding case. The distributions correspond to all $r=1,\ldots,100\text{ and } j=1,\ldots,5.$
(a) $\erreurprobarelative(1300,\Gamm^{r},j)$,  the relative probability error. 
(b) $\epsilon_P^{IS}(10^7,\Gamm^{r},j)$, the IS coefficient of variation of membership probability. 
(c) $\epsilon_{\Gamm}^{IS}(10^6,\Gamm^{r},j)$, the IS centroid standard deviation.
}
\end{figure*}

\paragraph{Flooding importance sampling error.}

We complete this study by an analysis of the importance sampling error as explained in \citealt{sire_is}. 
A bootstrap approach provides, for each of the $j=1,\ldots,\card$ centroids of each $\Gamm^r$, the IS coefficient of variation of the membership probability, $\epsilon_P^{IS}(\tilde{n},\Gamm^{r},j)$, and the IS centroid standard deviation, $\epsilon_{\Gamm}^{IS}(\ntrue,\Gamm^{r},j)$.
The distributions of these metrics are given in Figure ~\ref{std_flooding} and \ref{centroid_flooding} for $\tilde{n} = 10^7$ and $\ntrue = 10^6$. $\epsilon_P^{IS}(\tilde{n},\Gamm^{r},j)$ is lower than $1\%$, showing that the error on the probability is mainly driven by the error of the metamodel. 
$\epsilon_{\Gamm}^{IS}(\ntrue,\Gamm^{r},j)$ is less than $0.02$~meters for $\ntrue=10^6$ which is quite satisfying in a flooding context.

\paragraph{Flooding prototype maps.} 

The application of the Prototype Maps Algorithm to Les Boucholeurs with $\ntrue = 10^6$ gives the $5$ representative maps of Figure~\ref{rep_maps}.
Each of them represents a type of flooding. The probability masses are computed with $\tilde{n} = 10^7$ to increase precision, 
which was not possible in the Lloyd's algorithm because of the computational cost of the multiple iterations. 
The most serious the flooding, the lowest the probability, except for the Voronoi cell number 3  (Figure~\ref{gamma_star3}) which is a bit less likely than Voronoi cell number 4 (Figure~\ref{gamma_star4}), even if it is associated to more moderate floodings. 
Although this seems counter-intuitive,  it is possible that a specific type of small flooding  is less likely than another type of more severe flooding. 
Among the Voronoi cells associated to the strongest floodings, the prototype number 4 (Figure~\ref{gamma_star4}) is a map with shallow water covering a large area, while the flooded area is smaller for cell number 5 (Figure~\ref{gamma_star5}) but with a greater water depth near the coastline.
These 5 Voronoi cells can be visualized on the two first FPCA axes that summarize 98\% of the inertia of the 1300 flooding maps. 
Figure~\ref{fig-fpcaFloodVisu} illustrates this (and Appendix~\ref{axes_fpca} discusses these results more thoroughly), showing that the fourth and fifth cells associated to the most extreme floodings contains very different types of maps.

\begin{figure}[htp]
\centering

\begin{subfigure}{0.53\columnwidth}
\centering
\includegraphics[width=0.99\textwidth]{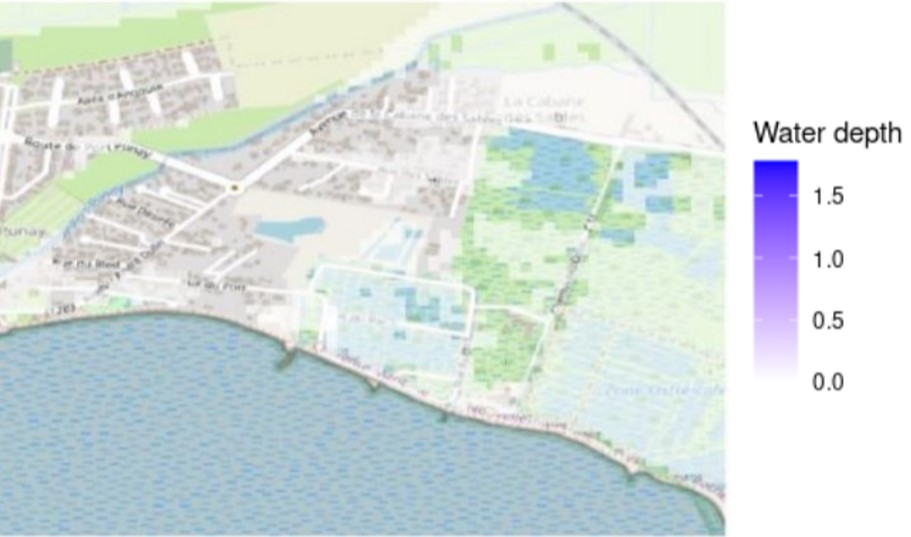}
\subcaption{\small$\hat{P}_{\tilde{n}}(\Gammapred,1,\Ypred) = 9.7\times10^{-1}$, 1 in 1.03}
\label{gamma_star1}
\end{subfigure}
\linebreak
\begin{subfigure}{0.43\columnwidth}
\centering
\includegraphics[width=0.99\textwidth]{gamma_star2.jpg}
\subcaption{\small$\hat{P}_{\tilde{n}}(\Gammapred,2,\Ypred) = 2.0\times10^{-2}$, 1 in 50}
\label{gamma_star2}
\end{subfigure}
\begin{subfigure}{0.43\columnwidth}
\centering
\includegraphics[width=0.99\textwidth]{gamma_star3.jpg}
\subcaption{\small$\hat{P}_{\tilde{n}}(\Gammapred,3,\Ypred) = 5.2\times10^{-3}$, 1 in 192}
\label{gamma_star3}
\end{subfigure}\hfill
\medskip
\begin{subfigure}{0.43\columnwidth}
\centering
\includegraphics[width=0.99\textwidth]{gamma_star4.jpg}\hfill
\subcaption{\small$\hat{P}_{\tilde{n}}(\Gammapred,4,\Ypred) = 6.2\times10^{-3}$, 1 in 161}
\label{gamma_star4}
\end{subfigure}
\begin{subfigure}{0.43\columnwidth}
\centering
\includegraphics[width=0.99\textwidth]{gamma_star5.jpg}
\subcaption{\small$\hat{P}_{\tilde{n}}(\Gammapred,5,\Ypred) = 1.4\times10^{-3}$, 1 in 714}
\label{gamma_star5}
\end{subfigure}
\caption{Centroids calculated by the Prototype Maps Algorithm with $\ntrue = 10^6$ predicted maps, and their probability mass estimated with $\tilde{n} = 10^7$. The frequency, i.e., the inverse of the probability is provided for an easier interpretation.
The water depths (in meters) are drawn with colors ranging from light to dark blue. Blue contour lines represent a water depth of 15~cm.}

\label{rep_maps}

\end{figure}

\section{Summary and perspectives}  \label{perspectives}

This work proposes a new visualization method based on quantization 
adapted to rare events and for which a limited database is available, as it occurs with computationally costly simulations.
Without loss of generality, the method, called Prototype Maps Algorithm, is developed around an application to flood maps.
The visualization takes the form of
a set of prototype maps and their probability masses
that best represent the continuous probability distribution on the space of maps. 
Importance sampling and a metamodel scheme based on Functional PCA combined with Gaussian process regression are used to 
address the challenges of the scarcity of the (flooding) events and the small data set size. 
 The obtained prototype maps along with their probabilities are an intuitive representation of the flooding risk while remaining meaningful mathematically. 
Indeed, the prototype maps are a discrete approximation of the continuous random flooding event. 
The distance between the discrete and the continuous distribution is the quantization error, 
which is easily interpretable as it is the expected squared distance between a random map and its closest prototype map.

Performance metrics used in the article validate the precision of our method on an analytical test problem and a real flooding case.

The following developments are worthwhile, 
the first being related to the tuning of the algorithm, the others dealing with complementary approaches.

%\subsection{Tuning of the algorithm}
\paragraph{Tuning of the algorithm.}

First, a more sophisticated biased density function of importance sampling could be sought. 
We have integrated importance sampling inside Lloyd's algorithm to deal with the low flooding probability when estimating the centroid and the probabilities of each cluster.
The choice of a uniform density $\biais$ proposed here for the offshore conditions is a compromise between simplicity and the exploration of the whole domain and was empirically seen to keep the variance of the estimators small.
Yet a real optimization of this density has not been performed. 
While it is a difficult problem with functional unknowns, the search for a better biased density, defined independently for each cell, is a relevant continuation to this work.

Second, the number of prototype maps, $\card$, is chosen a priori to remain usable in practice while capturing sufficiently diverse floodings.
    Strategies to identify the number of Voronoi cells could be implemented based, for instance, on the evolution of the quantization error with the number of representatives. 

\paragraph{Complementary approaches.}

The method described here characterizes representative floodings. 
It would now be valuable to know more about the forcing and breaching conditions leading to these scenarios, i.e., to study how the inputs are related to each prototype.

An inherent feature of the obtained quantization is that none of the prototype maps is the no-flood map. The centroid of the first Voronoi cell has a negligible water depth ($3\times10^{-5}$ meters on average), but it is not really empty. 
However, fortunately, it is known that the probability that a flood map is empty is close to 1. 
The formulation of quantization could then be adapted so as to force the no-flood map (or more generally a Dirac of a continuous probability law) to be one of the representatives.

In this article, the metamodel is the kriging mean. The uncertainties of the Gaussian process were not considered. 
Another future development could account for this metamodel uncertainty in the computation of the conditional expectations and the probabilities.

\FloatBarrier

\section*{Acknowledgement}

This research was conducted with the support of IRSN and BRGM, through the consortium in Applied Mathematics CIROQUO (\url{https://doi.org/10.5281/zenodo.6581217}), gathering partners in technological and academia in the development of advanced methods for Computer Experiments. 

\section*{Declaration of competing interest}

The authors declare that they have no known competing financial interests or personal relationships that could have appeared to
influence the work reported in this paper.

\nomenclature[01]{$\centro_{i}$}{Centroid of the Voronoi cell number $i$}
\nomenclature[02]{$\Gamm$}{A set of $\card$ prototype maps}
\nomenclature[02]{$\Gammaref$}{Optimal quantization obtained with the true maps (without metamodel)}
\nomenclature[03]{$\Gammapred$}{Optimal quantization obtained with the predicted maps}
\nomenclature[04]{$(\Gamm^{r})_{r\in \{1,\dots,n_{\Gamm}\}}$}{Family of $n_{\Gamm}$ prototype perturbations sampled to evaluate the errors}
\nomenclature[07]{$\hat{e}(\Gamm)$}{Empirical quantization error associated to $\Gamm$}
\nomenclature[08]{$\epsilon_{\Gamm}^{IS}(\ntrue,\Gamm^{r},j)$}{IS centroid standard deviation}
\nomenclature[09]{$\epsilon_P^{IS}(\tilde{n},\Gamm^{r},j)$}{IS coefficient of variation of the membership probability}
\nomenclature[10]{$\epsilon_{\Gamm}^{MM}$}{Excess in quantization error}
\nomenclature[11]{$\erreurprobarelative(n,\Gamm^{r},j)$}{Relative probability error}
\nomenclature[12]{$\estim_{n}(\Gamm,j,Y)$}{Estimator of the centroid of the $j$th Voronoi cell of $\Gamm$ computed with $n$ maps with the output function $Y$}
\nomenclature[12]{GP}{Gaussian Process}
\nomenclature[13]{$\biais$}{Biased density used in the importance sampling scheme}
\nomenclature[13]{$\card$}{Number of prototype maps}
\nomenclature[14]{$\tilde{n}$}{Number of maps sampled to compute the probabilities associated to the optimal quantization at the end of the Prototype Maps Algorithm}
\nomenclature[15]{$n_{\Gamm}$}{Number of prototypes perturbations to evaluate the errors metrics.}
\nomenclature[16]{$\nerr$}{Number of maps sampled to compute the empirical quantization error}
\nomenclature[17]{$n_{\mathrm{it}}$}{Number of iterations in the Prototype Maps Algorithm}
\nomenclature[18]{$\ntrue$}{Number of maps sampled for the Prototype Maps Algorithm iterations}
\nomenclature[19]{$\ntrain$}{Number of maps in the training database}
\nomenclature[20]{$\nvar$}{Number of maps sampled to compute the variance of the estimators (centroids and probabilities).}
\nomenclature[21]{$\hat{P}_{n}(\Gamm,j,Y)$}{Estimator of the probability mass of the $j$th Voronoi cell of $\Gamm$ computed with $n$ maps with the output function $Y$}
\nomenclature[22]{$s^2$}{Number of pixels of a map}
\nomenclature[22.2]{$\supp()$}{Support of a density}
\nomenclature[24.2]{$Y$}{Random output function introduced in the general formulation}
\nomenclature[23]{$\Dspace$}{Input space, for instance offshore and breaching conditions in the coastal case}
\nomenclature[23.2]{$X$}{Random inputs}
\nomenclature[23.3]{$\tilde X$}{Random inputs following the importance sampling density $\biais$}
\nomenclature[23.4]{$x$}{A vector of inputs, i.e. a realization of the random inputs}
\nomenclature[24.4]{$\Ytrue$}{True maps function (deterministic)}
\nomenclature[24.6]{$\Ypred$}{Predicted maps function (deterministic)}
\nomenclature[24]{$\Ecal$}{Output space of maps}
\nomenclature[05]{$\clust{\Gamm}{i}$}{Voronoi cell number $i$ associated to the quantization $\Gamm$}
\nomenclature[12]{FPCA}{Function Principal Components Analysis}
\nomenclature[13]{IS}{Importance Sampling}

\printnomenclature

\bigskip
\begin{center}
{\large\bf SUPPLEMENTARY MATERIAL}
\end{center}

\begin{description}

\item[Supplementary material description:] Note describing the git repository (readme.md).

\item[Codes related to the Campbell2D case:] Git repository containing R notebooks to reproduce all the experiments related to the Campbell2D function that are described in the article. (\url{https://github.com/charliesire/quantization_Campbell2D.git})

\item[R package FunQuant:] R package containing functions to perform the Prototype Maps Algorithm and compute the associated performance metrics. (\url{https://github.com/charliesire/FunQuant.git})

\item[Importance sampling performance metrics:] Note providing details about the computation of the importance sampling performance metrics.(\texttt{is\_perf\_metrics.pdf}).

\item[Complexity of the Prototype Maps Algorithm:] Note explaining the computation of the complexity of the Prototype Maps Algorithm (\texttt{complexity\_pma.pdf}).

\end{description}

\bibliographystyle{Chicago}

\bibliography{Quantizing_rare_random_maps}

\appendix

\section{Proof of theorem \ref{pointfixe}}\label{proof_theorem}

In this Appendix, we generalize the proof of optimal quantization of Kieffer and Cuesta-Albertos (\citealt{Pages2} ) to random fields $Y$.
Let us consider a probability space $(\Omega_{0}, \mathcal{F}_{0}, \mathbb{P})$, a random variable $Y \colon \Omega_{0} \to \Ecal$ and the optimal quantization of $\card$ prototype maps $\Gamm^{\star} = \{\centro^{\star}_{1},\dots, \centro^{\star}_{\card}\} \in \underset{\Gamm \in \Ecal^{\card}}{\argmin}\: (e(\Gamm))$.
%Let us consider a probability space $(\Omega_{0}, \mathcal{F}_{0}, \mathbb{P})$, a random variable $Y \colon \Omega_{0} \to \Ecal$ and the optimal quantization of $\card$ prototype maps $\Gamm^{\star} = \{\centro^{\star}_{1},\dots, \centro^{\star}_{\card}\} &\in \underset{\Gamm \in \Ecal^{\card}}{\argmin}\: (e(\Gamm))$.

We suppose that
\begin{itemize}
    \item $\Ecal$ is of finite dimension $m$~,
    \item $\forall z,y \in \Ecal,\scalprodE{z}{y} = \sum_{i = 1}^{m} \lambda_{i}z_{i}y_{i}$ with $\forall i, \lambda_{i} > 0$~,
    \item and $Y \in L^{2}_{\Ecal}(\mathcal{F}_{0}) =  \{W  \colon \Omega_{0} \to \Ecal \in \mathcal{F}_{0}, \Esp{\norm{W}^{2}} < +\infty\}$~.
\end{itemize}
Let us show that ~ $\forall i \in \{1\dots\card\}, \quad \Esp{Y\mid \yinclust{Y}{\Gamm^{\star}}{i}} = \centro^{\star}_{i}$.

We introduce the inner product of the space  $L^{2}_{\Ecal}(\mathcal{F}_{0})$: $\langle W,Y \rangle = \scalprodomega{W}{Y}, \: W,Y \in L^{2}_{\Ecal}(\mathcal{F}_{0})$.
 We set $\mathcal{F} = \sigma(\rep(Y)) \subset \mathcal{F}_{0}$ the $\sigma$-algebra induced by $\rep(Y)$.
\begin{proposition}\label{proj_ortho}
Let $W \in L^{2}_{\Ecal}(\mathcal{F}_{0})$, 
$\pi_{L^{2}(\mathcal{F})}(W) = \Esp{W\mid \mathcal{F}}$, with $\pi_{L^{2}(\mathcal{F})}(W)$ the orthogonal projection of $W$ on $L^{2}(\mathcal{F})$ (\citealt[pp. 229]{Durret})
\end{proposition}

\begin{proof}
Let $W \in L^{2}_{\Ecal}(\mathcal{F}_{0}). ~ \forall Z \in L^{2}_{\Ecal}(\mathcal{F}),$
\begin{equation*}
\begin{split}
\scalprod{Z}{W} & = \Esp{\Esp{\scalprodE{Z}{W} \mid \mathcal{F}}} \\
 & = \Esp{\Esp{\sum_{i = 1}^{n} \lambda_{i}z_{i}w_{i} \mid \mathcal{F}}} \\
 & = \Esp{\sum_{i = 1}^{n}\lambda_{i}\Esp{z_{i} w_{i} \mid \mathcal{F}}}\\
 & = \Esp{\sum_{i = 1}^{n}\lambda_{i}z_{i}\Esp{ w_{i} \mid \mathcal{F}}} \: \mbox{because } \forall i, \: z_{i} \in \mathcal{F} \\
 & = \Esp{\scalprodE{Z}{\Esp{W \mid \mathcal{F}}}}~.\\
\end{split}
\end{equation*}
\end{proof}
We have $\qgs(Y) - \Esp{Y \mid \qgs(Y)}\in \mathcal{F} =  \sigma(\qgs(Y))$, 
so $Y - \Esp{Y\mid \qgs(Y)}$ is orthogonal to $\qgs(Y) - \Esp{Y \mid \qgs(Y)}$ for the inner product of the space $L^{2}_{\Ecal}(\mathcal{F}_{0})$ according to proposition \ref{proj_ortho}.
Then, the Pythagorean theorem gives
\begin{equation}
\begin{aligned}
\Esp{\norm{Y - \qgs(Y)}^{2}} = \:&\Esp{\norm{Y - \Esp{Y\mid \qgs(Y)}}^{2}} + \\ &\Esp{\norm{\qgs(Y) - \Esp{Y \mid \qgs(Y)}}^{2}}\label{Pytha}~.
\end{aligned}
\end{equation}

We set $\tilde{\Gamm} = \{\tilde{\centro}_{1},\dots,\tilde{\centro}_{\card}\}$ such that $\forall i \in \{1\dots N\}, \tilde{\centro}_{i} = \Esp{Y\mid \yinclust{Y}{\Gamm^{\star}}{i}}$. 
\begin{itemize}
    \item $\Esp{\norm{Y - q_{\tilde{\Gamm}}(Y)}^{2}} \leq \Esp{\norm{Y - \Esp{Y\mid \qgs(Y)}}^{2}}$ : 
    \begin{align*}\forall y \in \Ecal, \exists i \in \{1,\dots,N\}, \norm{y - \Esp{Y\mid \qgs(Y) = \qgs(y)}} &= \norm{y - \Esp{Y\mid \yinclust{Y}{\Gamm^{\star}}{i}}}\\ &= \norm{y - \tilde{\centro}_{i}} \\ &\geq \norm{y - q_{\tilde{\Gamm}}(y)} \\ &\text{ by definition of } q_{\tilde{\Gamm}}.
    \end{align*}
    \item $\Esp{\norm{Y - q_{\tilde{\Gamm}}(Y)}^{2}} \geq \Esp{\norm{Y - \qgs(Y)}^{2}}$ by definition of $\Gamm^{\star}$
\end{itemize}
Then, $\Esp{\norm{Y - \qgs(Y)}^{2}} \leq \Esp{\norm{Y - \Esp{Y\mid \qgs(Y)}}^{2}}$.
\\
We conclude from Equation \eqref{Pytha}: 
$\Esp{\norm{\qgs(Y) - \Esp{Y \mid \qgs(Y)}}^{2}} = 0$ so that \linebreak $\qgs(Y) = \Esp{Y \mid \qgs(Y)}$.
\\
This is a known result (\citealt{Pages}) that, therefore, applies to our study with the random field $Y(X)$.

\section{Projection on a basis of maps}

\subsection{Functional Principal Components Analysis}\label{fpca_annex}

Every flood map $y(x)$ that we want to model can be considered as a matrix indicating the water depth at each pixel, i.e., as a function 
\begin{align*}
y_{x}\colon \mathbb{R}^{2} &\longrightarrow \mathbb{R} \\
z = (z_{1},z_{2}) &\longmapsto y_{x}(z).
\end{align*}

We work with a training database of $\ntrain$ maps $(y_{x^{i}}(z))_{i=1,\dots,\ntrain}$ and the objective here is to compute the projection of the flood maps in a FPCA basis of $\npc$ maps.
We first consider a functional orthonormal basis $\Phi(z) = (\Phi_{1}({z}),\dots,\Phi_{K}({z}))^{\top}$. D4 Daubechies wavelets (\citealt{Daubechies}), which are widely used in image processing, are chosen here, as \citealt{Perrin} shows their accuracy for a similar problem. The procedure consists in the following steps:  

\begin{enumerate}
    \item Find the coordinates $(\alpha(x^{i}))_{i=1,\dots,\ntrain}$ of the training maps $(y_{x^{i}}(z))_{i=1,\dots,\ntrain}$ in the $\Phi$ basis: \\
    $y_{x^{i}}(z) = \sum_{k=1}^{K} \alpha_{k}(x^{i})\Phi_{k}(z) = \Phi(z)^{\top} \alpha(x^{i})$ 
    \item Compute $\lambda_{k} = \Esp{\frac{\alpha_{k}(X)^2}{\sum_{j=1}^{K}\alpha_{j}(X)^2}}, \: k =  1\dots K$, order them by decreasing values
and select the $\tilde{K} \ll K$ most important vectors of the basis (the $\tilde{K}$ vectors associated with the highest $\lambda_{k}$) such that $\sum_{k=1}^{\tilde{K}} \lambda_{k} \geq p$, with $p$ a given proportion of energy. Then, we have $y_{x^{i}}(z) \simeq \sum_{k=1}^{\tilde{K}} \alpha_{k}(x^{i})\Phi_{k}(z) = \tilde{\Phi}(z)^{\top} \tilde{\alpha}(x^{i}) $
    \item Apply PCA in $\mathbb{R}^{\tilde{K}}$ on the dataset of coefficients evaluated at the design points $(\tilde{\alpha}(x^{i}))_{i=1,\dots,\ntrain}$. 
Denote $t_{1}(x^{i}),\dots,t_{\npc}(x^{i}), i = 1,\dots,\ntrain$, the coordinates of the first $\npc$ principal components and $\omega_{1},\dots,\omega_{\npc}$ the associated eigenvectors in $\mathbb{R}^{\tilde{K}}.$
We have $\tilde{\alpha}(x^{i}) \simeq \sum_{j=1}^{\npc} t_{j}(x^{i})w_{j} = \Omega^{\top}t(x^{i})$, with $\Omega$ the $\npc \times \tilde{K}$ projection matrix.
\end{enumerate}

Finally we can write $y_{x^{i}}(z) \simeq  (\Omega\tilde{\Phi}(z))^{\top}t(x^{i})$, with $\Omega\tilde{\Phi}$ a basis of $\npc$ maps. $\tilde{K}$ (or, equivalently, $p$) and $\npc$ are tuned by cross-validation.

\subsection{Inverse FPCA}\label{inverse_fpca}

As described in Section~\ref{Metamodelisation}, Gaussian processes are used to predict $\hat{t}(x^{\star})$ for a new input $x^{\star}$. 
The following steps provide the map $\Ypred(x^{\star})$ from the coordinates $\hat{t}(x^{\star})$:  
\begin{enumerate}
    \item Compute the coefficients of the vector $\alpha(x^{\star})$
    \begin{itemize}
        \item for $k = \tilde{K} + 1,\dots, K, \hat{\alpha}_{k}(x^{\star}) = \frac{1}{\ntrain}\sum_{i=1}^{\ntrain}\alpha_{k}(x^{i})$
        \item for $k = 1,\dots, \tilde{K}, \hat{\alpha}_{k}(x^{\star}) =  \sum_{j=1}^{\npc}\hat{t}_{j}(x^{\star})\omega_{j}$
    \end{itemize}
    \item Compute $(\hat{y}_{x^{\star}}(z))$ with $\hat{y}_{x^{\star}}(z) =  \sum_{k=1}^{K} \hat{\alpha}_{k}(x^{\star})\Phi_{k}(z)$
   
\end{enumerate}

 The computation of $y_{x^{\star}}(z)$ at the targeted family of pixels provides the predicted map $\Ypred(x^{\star})$.
 
\section{Probability Distributions}
\label{probas_inputs}
\newcolumntype{P}[1]{>{\centering\arraybackslash}p{#1}}

The density functions of the offshore variables are obtained through empirical observations as described in Section \ref{sec:flood_description}. 
They are shown in Table \ref{table_densities}, along with their support.
The location of the breach is an integer in $\{1,\dots,10\}$, and the erosion rate is in $[0,1]$.

\begin{table}
\centering
\begin{tabular}{ |P{1.8cm}|P{7cm}| P{1.7cm}|P{4.8cm}|}  
 \hline
 {\footnotesize Variable $X_i$} & {\footnotesize Distribution $I_i$} & {\footnotesize Support} & {\footnotesize Law type}\\
 \hline \hline
 \begin{tabular}{l}
\shortstack{$T$ \\ \color{white}{a}\color{black} \\ {\footnotesize Tide}} 
\end{tabular}
& \begin{minipage}{0.4\textwidth}
\color{white}{a}\color{black}\\
      \includegraphics[width=1\textwidth]{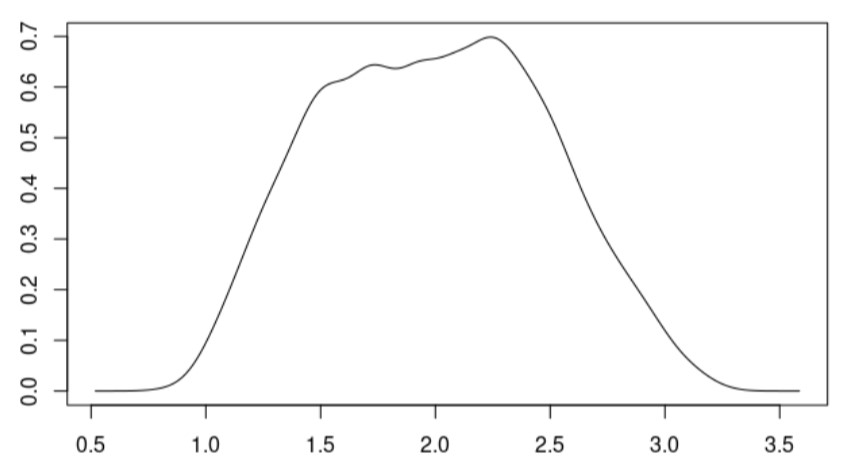}
    \end{minipage}
    & {\footnotesize [0.52,3.59]} & Empirical
\\ 
 \hline
\begin{tabular}{l}
\shortstack{$S$ \\ \color{white}{a}\color{black} \\ {\footnotesize Surge}} \end{tabular}
& \begin{minipage}{0.4\textwidth}

\color{white}{a}\color{black}\\
      \includegraphics[width=1\textwidth]{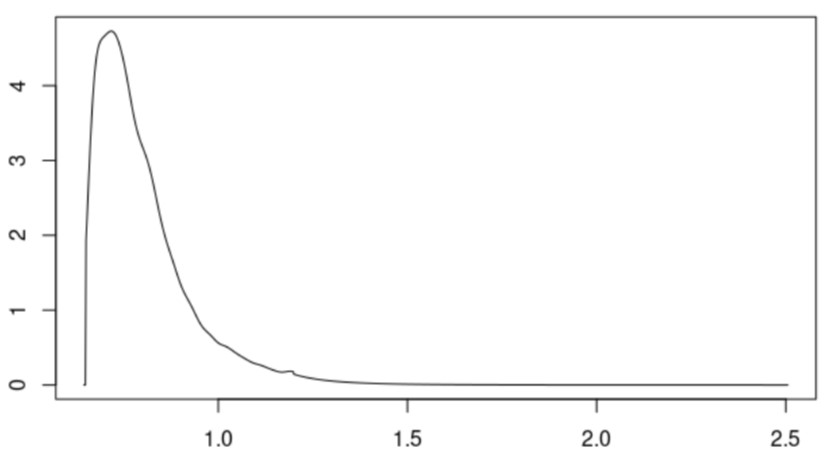}
    \end{minipage} 
    & {\footnotesize[0.65,2.5]} &   \begin{tabular}{l}
{\small Combination of a genera-} \\ {\small lised Pareto distribution} \\ {\small above a threshold of 0.55m} \\ {\small and the empirical distribu-}\\{\small tion derived from the tide}\\ gauge measurements
  \end{tabular}\\
  \hline
  \begin{tabular}{l}
\shortstack{$t_{0}$ \\ \color{white}{a}\color{black} \\ {\footnotesize Phase} \\ {\footnotesize difference}}
\end{tabular} & \begin{minipage}{0.4\textwidth}
 \color{white}{a}\color{black}\\
      \includegraphics[width=1\textwidth]{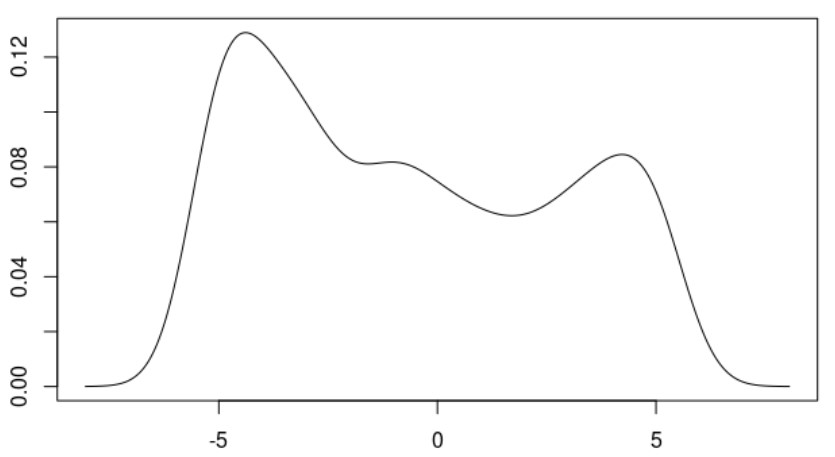}
    \end{minipage} & {\footnotesize[-8.05,8.05]} & Empirical \\
  \hline
 \begin{tabular}{l}
\shortstack{$t_{-}$ \\ \color{white}{a}\color{black} \\ {\footnotesize Duration} \\ {\footnotesize rising} \\ {\footnotesize part}}
\end{tabular}
&\begin{minipage}{0.4\textwidth}
\color{white}{a}\color{black}\\
      \includegraphics[width=1\textwidth]{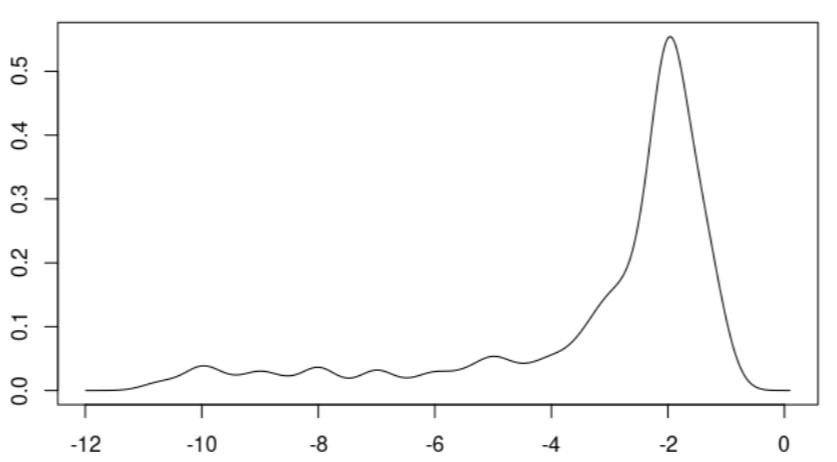}
    \end{minipage} & {\footnotesize[-12,0]} & Empirical\\
  \hline
\begin{tabular}{l}
    \shortstack{$t_{+}$ \\ \color{white}{a}\color{black} \\ {\footnotesize Duration} \\ {\footnotesize falling} \\ {\footnotesize part}}
  \end{tabular} & \begin{minipage}{0.4\textwidth}
      \color{white}{a}\color{black}\\\includegraphics[width=1\textwidth]{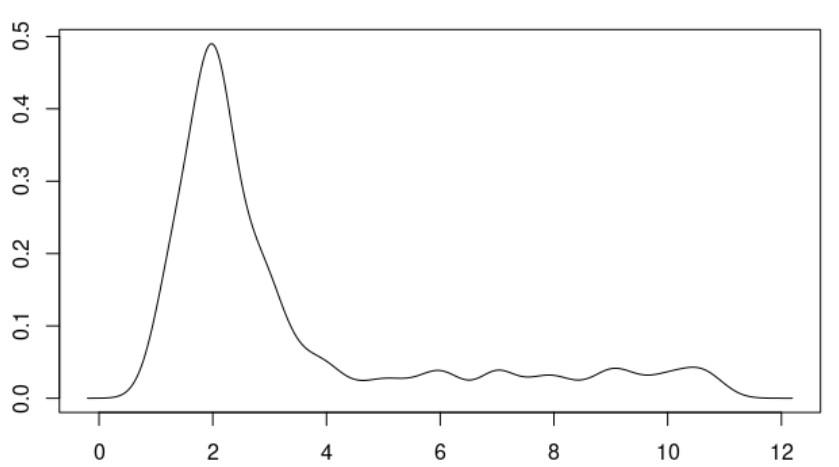}
    \end{minipage} & {\footnotesize[0,12.2]} & Empirical\\
\hline
\end{tabular}
\caption{Variables describing the offshore conditions and their probability distribution.
}
\label{table_densities}
\end{table}

The affine transformation $a$ from $\supp(f_{X})$ to $[-1,5]^{7}$, introduced in the Campbell case, is then
\begin{align*}
a \colon  \supp(f_{X}) &\to [-1,5]^{7}\\
x &\mapsto
\begin{bmatrix}
    -1+6\frac{x-0.52}{3.59-0.52} \\ -1+6\frac{x-0.65}{2.5-0.65}\\ -1+6\frac{x-8.05}{8.05+8.05}\\
    -1+6\frac{x+12}{12}\\
    -1+6\frac{x-12.2}{12.2}\\
    -1+6\frac{x-1}{9}\\
    -1+6x
\end{bmatrix}
~.
\end{align*}

\section{Voronoi cells on the FPCA axes}\label{axes_fpca}

In the flooding case, the two first FPCA axes summarize $98\%$ of the inertia of the training 1300 maps. It is relevant to visualize the quantization cells through their two first FPCA coefficients. Among other things, this allows to have a better understanding of their size. 
Figure~\ref{pca_clusters_zoom} shows the distribution of $10^5$ maps sampled with the uniform density $\biais$ 
and the prototype maps in the FPCA basis calculated for the maps without breach. 
As expected, cell 1 is almost a Dirac at the empty map whose coefficients are $(-27,8)$. 
Moreover, the largest cells are the ones associated to the most extreme floodings. For instance, cell 5 contains very diversified maps in terms of water depth magnitude and spatial distribution. It is an expected property of the optimal Voronoi cells: the cells with high probabilities are small and vice versa. However, this characteristic may not be true is some specific contexts, as shown in Appendix~\ref{counter_example}.
Figure~\ref{pca_ponderation_zoom} shows the probability density of the maps expressed in the 2 first FPCA coefficients. 
It is a histogram where each bin has a colour corresponding to the weight $\log\left(\sum_{k=1}^{n_\text{bin}} \frac{f(\tilde X^k)}{\biais(\tilde X^k)}\right)$ associated to the IS density ratio of the maps, $n_\text{bin}$ being the number of maps falling within each bin. It is essential to understand the position of the prototype maps and the probability mass of the cells. 
It illustrates the concentration of the mass in cell 1, 
with a high $\log\left(\sum_{k=1}^{n_{\text{bin}}} \frac{f(X^{k})}{\biais(X^{k})}\right)$ compared to the other cells. 
It also explains the position of the centroid of the cell 5, which is located near the left limit. This is due to the low weight terms at the right of the cell.

\begin{figure}
\centering
  \begin{subfigure}[a]{0.67\textwidth}
\includegraphics[width=1\linewidth, height = 0.7\linewidth]{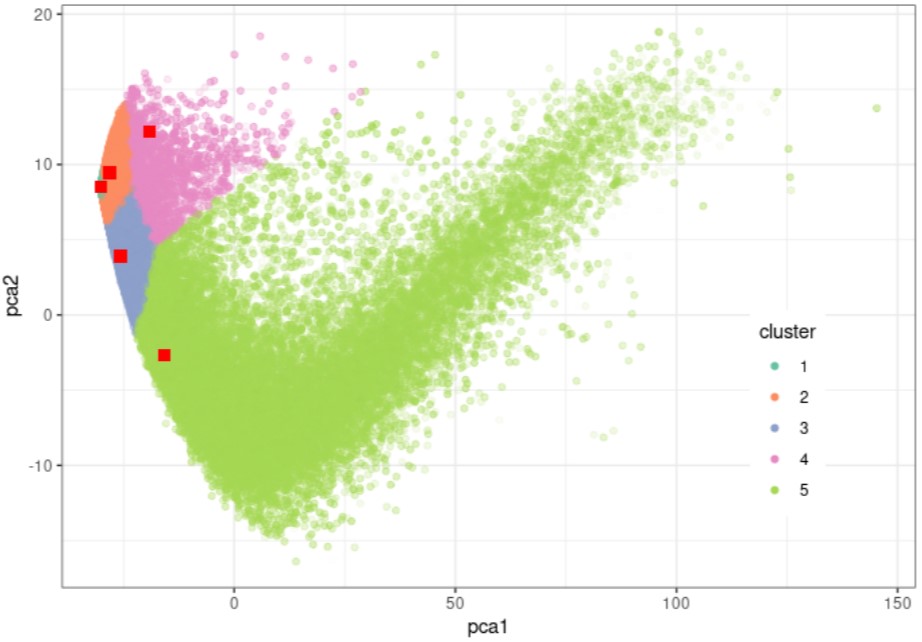}
\caption{}
  \label{pca_clusters_zoom}

  \end{subfigure}
  \begin{subfigure}[b]{0.67\textwidth}
\includegraphics[width=1\linewidth]{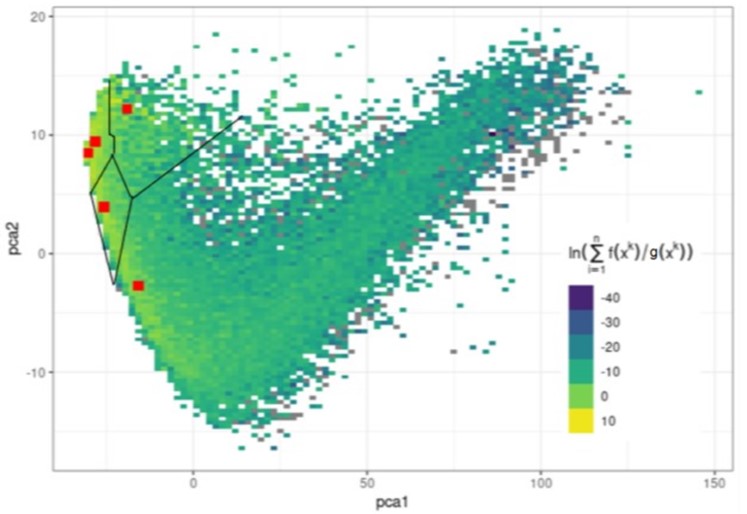}
\caption{}
\label{pca_ponderation_zoom}
  \end{subfigure}
  \caption{
(a) Distribution of $10^5$ points sampled with density $\biais$ in the different Voronoi cells in the two first FPCA axes of maps without breach. 
(b) 2D histogram where for each bin, the quantity 
$log\left(\sum_{k=1}^{n_\text{bin}} \frac{f(X^{k})}{\biais(X^{k})}\right)$ 
is displayed with the color scale. The prototype maps are represented by the red squares. The limits of cells 2, 3 and 4 are drawn with black lines. 
Grey pixels are associated to a null weight $\left(\sum_{k=1}^{n_\text{bin}} \frac{f(X^{k})}{\biais(X^{k})}\right)$.
\label{fig-fpcaFloodVisu}
}
\end{figure}

\section{Relation between cells size and probability}
\label{counter_example}

As the Lloyd's algorithm minimizes the quantization error $e(\Gamm) = \left[\Esp{\norm{Y-\rep(Y)}^{2}}\right]^{\frac{1}{2}}$, we can expect that the largest Voronoi cells are associated to a small probability mass. However, it may not be the case with specific random variables $Y$.

For instance, let consider a random variable $Y\in [-20,20]$, with a density function $d = \frac{1}{10}\mathds{1}_{[-20,-10]} + \frac{9}{10}\mathds{1}_{[0,20]}$. 

If we perform the Lloyd's algorithm with two centroids $\gamma_{1}$ and $\gamma_{2}$, it gives $\gamma_{1} = -15$ and $\gamma_{2} = 10$ and two Voronoi cells: $C_{1} = [-20,-2.5]$ with probability mass $p_{1} = \frac{1}{10}$ and $C_{2} = [-2.5,20]$ with probability mass $p_{2} = \frac{9}{10}$.
\end{document}